\documentclass[12pt,letterpaper]{article}
\usepackage{jheppub}
\usepackage[utf8]{inputenc}
\usepackage{amsthm}
\usepackage{dcolumn}
\usepackage{cancel}
\usepackage{booktabs}
\usepackage{multirow}
\usepackage{esvect}

\newcolumntype{C}[1]{>{\centering\let\newline\\\arraybackslash\hspace{0pt}}m{#1}}
\newtheorem{lemma}{Lemma}
\newtheorem*{corollary}{Corollary}
\def\be{\begin{equation}}
\def\ee{\end{equation}}
\def\ba#1\ea{\begin{align}#1\end{align}}
\def\bg#1\eg{\begin{gather}#1\end{gather}}
\def\bm#1\em{\begin{multline}#1\end{multline}}
\def\bmd#1\emd{\begin{multlined}#1\end{multlined}}
\def\rrb#1{\raisebox{0cm}{\rotatebox{-45}{#1}}}

\def\a{\alpha}
\def\b{\beta}

\def\s{\sigma}

\def\la{\label}

\def\re{\ref}
\def\er{\eqref}
\def\se{\section}
\def\sse{\subsection}

\def\td{\tilde}

\def\cd{\cdots}

\def\qu{\quad}

\def\({\left(}
\def\){\right)}
\def\[{\left[}
\def\]{\right]}
\def\<{\langle}
\def\>{\rangle}

\begin{document}

\title{Holographic Entropy Cone \mbox{with Time Dependence} in Two Dimensions}
\author[a]{Bart{\l}omiej Czech}
\emailAdd{bartlomiej.czech@gmail.com}
\affiliation[a]{Institute for Advanced Study, Tsinghua University, Beijing 100084, China}
\author[b]{and Xi Dong}
\emailAdd{xidong@ucsb.edu}
\affiliation[b]{Department of Physics, University of California, Santa Barbara, California 93106, USA}

\abstract{
In holographic duality, if a boundary state has a geometric description that realizes the Ryu-Takayanagi proposal then its entanglement entropies must obey certain inequalities that together define the so-called holographic entropy cone. A large family of such inequalities have been proven under the assumption that the bulk geometry is static, using a method involving contraction maps. By using kinematic space techniques, we show that in two boundary (three bulk) dimensions, all entropy inequalities that can be proven in the static case by contraction maps must also hold in holographic states with time dependence.
}

\maketitle

\se{Introduction}

\noindent
There is a widespread belief that the Ryu-Takayanagi (RT) proposal \cite{rt, rt2} is a powerful hint for understanding quantum gravity. It says that in holographic duality von Neumann entropies of boundary subregions are `geometrized' in the bulk: they are represented by extremal surfaces (when the bulk theory is Einstein gravity) or similar extremal, extended objects. What does this fact tell us about the fundamental theory of gravity?

Part of the message has already been decoded; examples include subregion-subregion duality \cite{dualofrho, subregionduality} and the quantum error-correcting property of the bulk \cite{errorcorr}. Yet we expect the RT formula---and its covariant generalization, the Hubeny-Rangamani-Takayanagi (HRT) formula \cite{hrt, maximin}---to usher further progress. One promising direction to explore is the following question: which theories and states admit classical bulk duals? This problem (and its converse \cite{inviolable}) has inspired several lines of ongoing research, for example \cite{jamiejoe, mellin, adsfromcft, erepr, universalspectrum, bulkpoint, changlin, entholonomies, kabatlifschytz}. The present paper is an amalgamation of two approaches, which investigate special properties of holographic states and the implications of the RT and HRT formulas.

The program \cite{hec} that directly begot this paper starts by asking: if entanglement entropies are `geometrized' as prescribed by Ryu and Takayanagi, what restrictions on the class of states does this impose? A hallmark example of extra restrictions levied by the geometrization of entanglement is the monogamy of mutual information \cite{monogamy}:
\begin{equation}
S({AB})+S({BC})+S({CA}) \ge S({ABC})+S(A)+S(B)+S(C).
\label{mmiintro}
\end{equation}
Here $A, B, C$ are disjoint boundary regions and $S(X)$ is the von Neumann entropy on $X$ which can be one of these regions and their various unions.
To formalize the problem, one considers the set of all possible tuples of entanglement entropies on these regions and their unions that can be achieved by a quantum state such that each individual entropy is an extremum of some functional such as a bulk area. A simple argument reveals that this object is a cone in the multi-dimensional `entropy space' whose axes parameterize the von Neumann entropies of subregions and their unions. The shape of the cone demarcates the divide between states with smooth geometric descriptions and other states, the latter being either absent or very exotic from the viewpoint of semiclassical gravity.

Any hypersurface in the entropy space that does not intersect the holographic entropy cone defines an inequality that is obeyed by all states with semiclassical bulk duals. Of particular interest are inequalities represented by hypersurfaces that are tangent to the cone; their joint envelope is the boundary of the entropy cone. If the cone is polyhedral then it is fully determined by a finite set of inequalities, which are identified with facets of the cone and which we now focus on. Inequality~(\ref{mmiintro}), the monogamy of mutual information, is an example of such a facet. Saturating monogamy or any other holographic inequality identified with a facet isolates an interesting class of states: ones whose entanglement entropies are only marginally amenable to an RT-style geometrization, in the sense that a small deformation of their entanglement entropies could make it impossible to view them as extremal values of some bulk functional. In this way, every facet of the holographic entropy cone (every inequality) presumably reflects some essential property of states that can describe a smooth geometry in quantum gravity.

Little is known for certain about the full holographic entropy cone. However, if we impose the additional condition that the conformal field theory (CFT) states and their dual bulk geometries be static\footnote{By `static' we mean that there is a time reflection symmetry with respect to a time slice on which the boundary regions lie.}, the \emph{static holographic entropy cone} has been characterized in great detail \cite{hec}. Its authors found infinitely many inequalities obeyed by static holographic states which, in particular, include \emph{all} such inequalities for up to $n=5$ named regions $A, B, C, D, E$. This means that the `$n \leq 5$ {static} holographic entropy cone' has been fully determined.\footnote{The $n=5$ inequalities proven in \cite{hec} were only later shown to form a complete set \cite{Cuenca:2019uzx}.} Our goal in this paper is to check whether time-dependent holographic states can violate any of the static entropy inequalities, including those proven in \cite{hec}. Restricted to the context of fewer than six named regions, our question is this: is the full holographic entropy cone larger than the static one?

There is a second lesson drawn from the RT proposal, which also undergirds the material in this paper: that understanding bulk physics can be simplified by using bi-local quantities on the boundary. Two examples of this are the bit thread prescription for computing entanglement entropies \cite{bitthreads, bitthreads2} and the kinematic program \cite{intgeometry}. The former converts the task of finding entanglement entropies into a problem of maximizing a certain flow; the integral curves of the flow have two endpoints on the boundary, na{\"\i}vely interpretable as the locations of degrees of freedom tied by a bi-local correlation.\footnote{\baselineskip=12pt It would be interesting to understand all holographic inequalities found in \cite{hec} in the language of bit threads. For the monogamy of mutual information (\ref{mmiintro}), this was carried out in \cite{bitthreadmonogamy, coopflows}.} The kinematic program, in turn, seeks to organize the data about the CFT and AdS using bi-local objects such as bulk geodesics or OPE expansions of pairs of operators \cite{diffentropy, entwinement, stereoscopic, robjanmichal} (see also \cite{geodesicwitten}). In the present paper, we use kinematic space to organize the data about entanglement entropies (computed by the HRT proposal) in time-dependent settings. 

By doing so, we find that in two boundary (three bulk) dimensions, all entropy inequalities that can be proven in the static case by contraction maps (including those proven in \cite{hec}) continue to hold in time-dependent states. In particular, this means that in two boundary dimensions time-dependent states do not make (the known part of) the holographic entropy cone larger than the static cone. On a technical level, going to kinematic space allows us to retain all the truly essential ingredients of the proof in \cite{hec} while disposing of the assumption of a static bulk. Where the authors of \cite{hec} cut and glued subregions of the bulk, we manipulate more abstract quantities defined in kinematic space, which generalize the static concept of the intersection number of two geodesics.\footnote{\baselineskip=12pt Ref.~\cite{sepintime} explains that our kinematic space `intersection numbers' can be related to ordinary, static intersection numbers by modular flow.} More conceptual comments on why kinematic space is helpful in generalizing the proof of \cite{hec}---and why this benefit is limited to two boundary dimensions---are contained in the Discussion.

The paper is organized as follows: In Sec.~\ref{prelims} we set up the problem and explain the necessary machinery from kinematic space. Sec.~\ref{staticreview} reviews the static proof of holographic entropy inequalities given in \cite{hec}. Sec.~\ref{strategy} proves the same inequalities in time-dependent pure states on a circle or line; this is the main result of our paper. Because our proof may be challenging to parse in a casual reading, we illustrate it with an informative example in Sec.~\ref{examples}. Sec.~\ref{nonpure} extends the proof of Sec.~\ref{strategy} to mixed states and CFTs on other (disconnected) topologies. We comment on the significance and outlook of our results in the Discussion.

\se{Preliminaries}
\label{prelims}
\subsection{Notation}

Between now and Sec.~\ref{nonpure} we will assume that the two-dimensional CFT is in a pure state and lives on a connected Lorentzian manifold---either Minkowski space $\mathbb{R}^{1,1}$ or a cylinder $S^1 \times {\rm time}$. The extension of our proof to CFTs living on disjoint unions of such manifolds and to mixed states---as is the case in holographic duals of multi-boundary black holes---is covered in Sec.~\ref{nonpure}.

We will be proving inequalities of the form
\begin{equation}
\sum_{l=1}^L \alpha_l S(I_l) 
\geq 
\sum_{r=1}^R \beta_r S(J_r)\,,
\label{templateineq}
\end{equation}
where $\alpha_l$ and $\beta_r$ are positive coefficients. The $I_l$ and $J_r$ are subregions on some space-like slice of the CFT and we do not assume that they are connected.

It is useful to set a notation for the connected components of the regions $I_l$ and $J_r$. We will refer to such connected components as $X_i$, with $i$ indexing the ordering of the intervals on the CFT slice. (The ordering is the reason why it is convenient to assume that the CFT lives on a connected manifold.) When the CFT lives on a circle, $i$ is understood modulo $N$, where $N$ is the total number of disjoint intervals that comprise the $I_l$s and $J_r$s. The interval which separates $X_{i-1}$ from $X_i$ will be called $Y_i$; when $X_{i-1}$ and $X_i$ are contiguous, $Y_i=\emptyset$. As a final piece of notation, we let:
\begin{equation}
Z_{2i} = X_i \qquad {\rm and} \qquad Z_{2i-1} = Y_i.
\label{defzxy}
\end{equation}
The indices of the $Z$s are valued modulo $2N$.

If we vary the relative sizes of the $Z$-intervals, the holographic entanglement entropies will undergo phase transitions. This happens when different collections of geodesics that connect interval endpoints exchange dominance and become minimal as stipulated by the RT proposal. A prototypical example of this phenomenon is the phase transition in the holographic entanglement entropy of two disjoint intervals, which was studied in \cite{phasetrans}. This entanglement entropy can be either in the connected or disconnected phase; see Fig.~\ref{fig:phases}. Our proof will require tracking the phases of the terms on the left hand side of (\ref{templateineq}) and adjusting the phases of terms on the right hand side. Thus, it is important to have an efficient vocabulary for identifying and distinguishing such phases. We will refer to distinct phases as colorings. 

\begin{figure}
        \centering
        \includegraphics[width=0.38\textwidth]{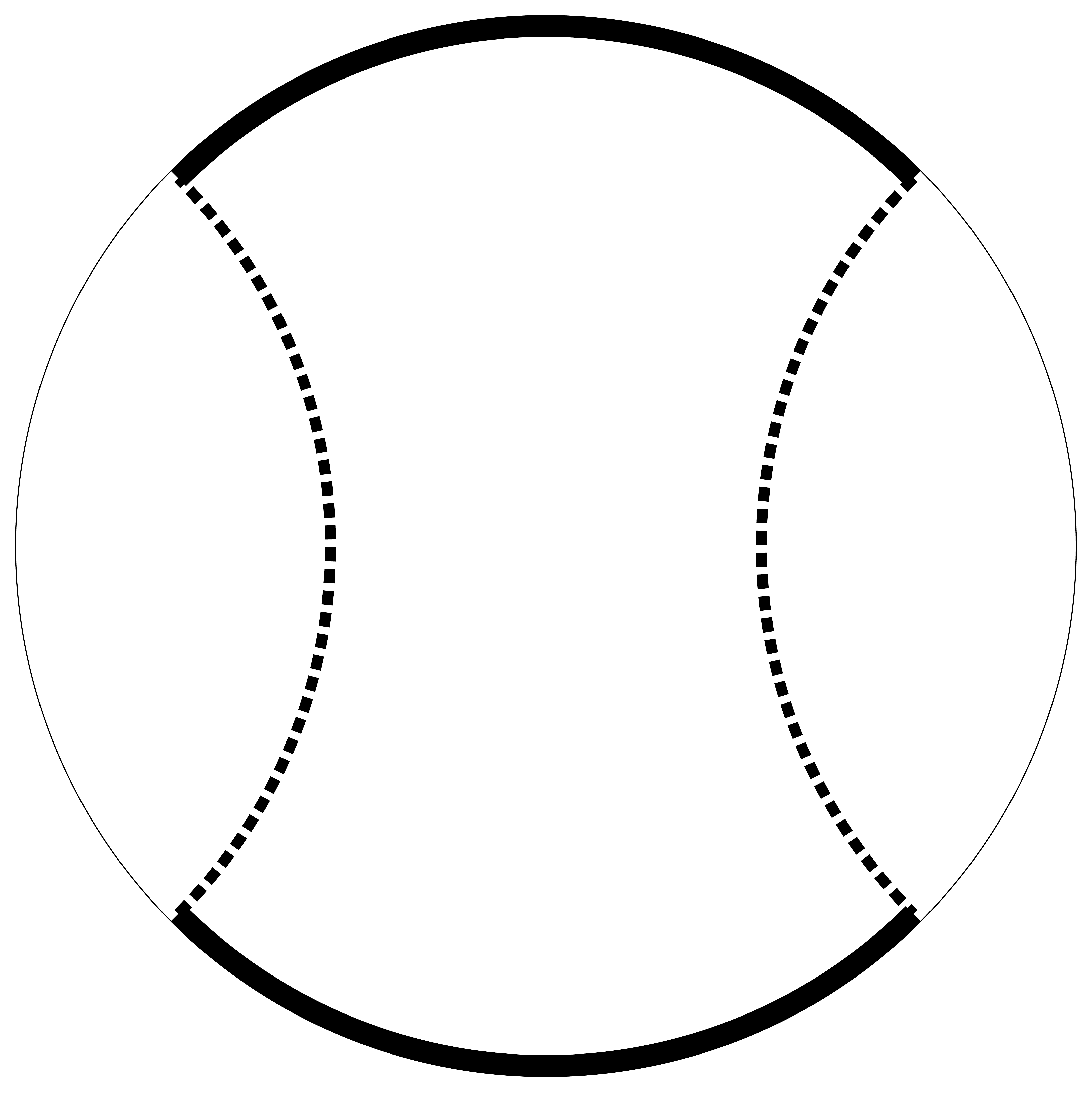}
        \hfill
        \includegraphics[width=0.38\textwidth]{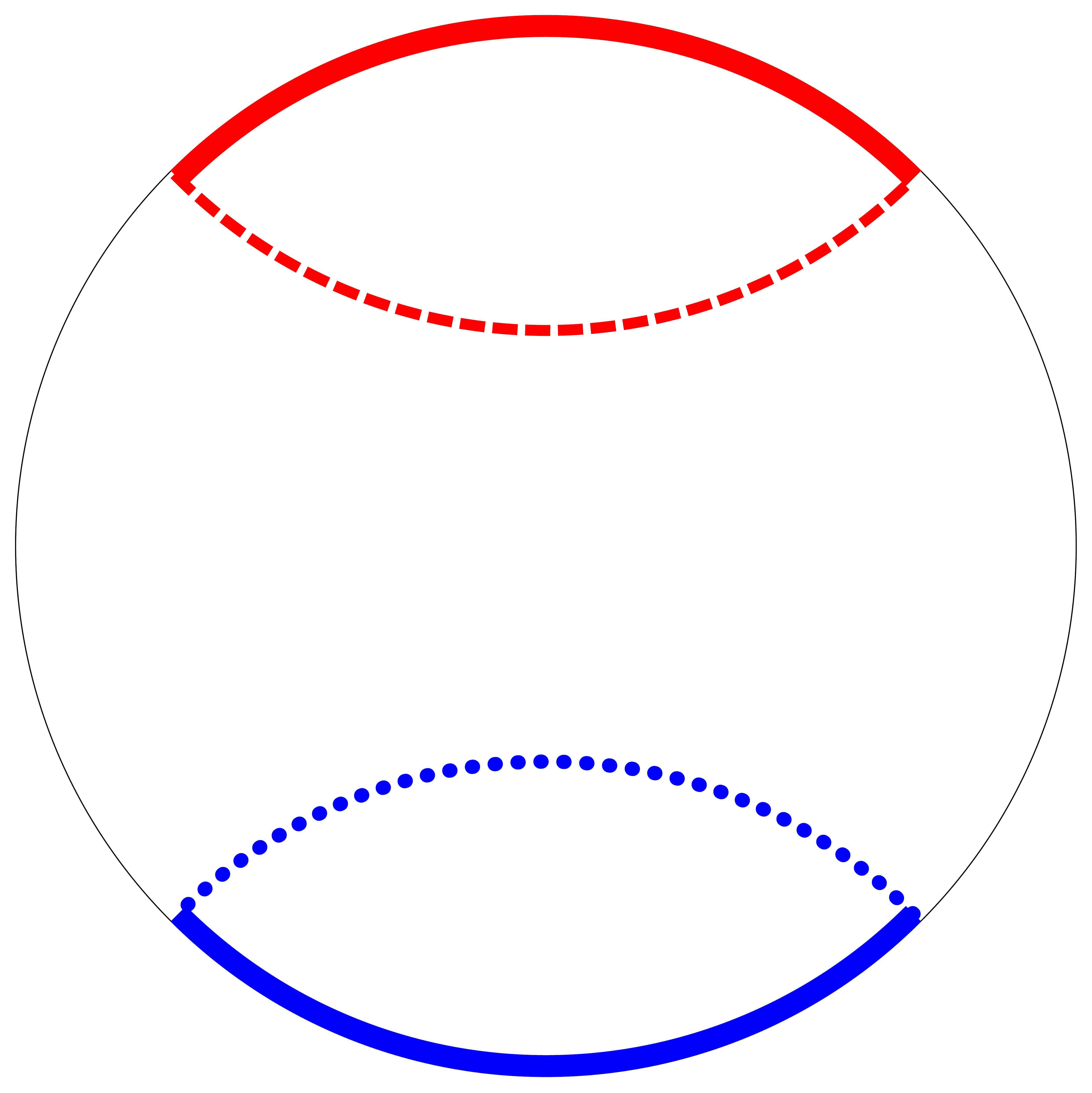}
        \caption{The connected (left) and disconnected (right) phases of the holographic entanglement entropy of two intervals. The coloring of the disconnected phase is explained in the text; see also Fig.~\ref{fig:example}.}
        \label{fig:phases}
\end{figure}

Consider some $S(I_l)$ (or $S(J_r)$), which is part of inequality~(\ref{templateineq}). Region $I_l$ is the union of some number of $X_i$s. We will partition these connected components of $I_l$ into colors. A coloring of the components of $I_l$ specifies a phase of $S(I_l)$ in that all intervals marked with the same color (and only they) are connected in that phase. For example, if $I_l$ comprises four intervals $X_1, X_2, X_3, X_4$, we have the following colorings:
\begin{align}
& (X_1 X_2 X_3 X_4)~{\rm or} 
\nonumber \\ &
(X_1) (X_2 X_3 X_4)~{\rm or}~(X_2) (X_1 X_3 X_4)~{\rm or}~(X_3) (X_1 X_2 X_4)~{\rm or}~(X_4) (X_1 X_2 X_3)
\nonumber \\ &
(X_1 X_2) (X_3 X_4)~{\rm or}~\xcancel{(X_1 X_3)(X_2 X_4)}~{\rm or}~(X_1 X_4) (X_2 X_3)~{\rm or}
\nonumber \\ &
(X_1) (X_2) (X_3 X_4)~\textrm{and 5 other permutations or}
\nonumber \\ &
(X_1) (X_2) (X_3) (X_4).
\label{phasescycles}
\end{align}
The top option is the completely connected phase; the bottom one is the completely disconnected phase. The left panel of Fig.~\ref{fig:example} depicts an example coloring (phase) of a seven-interval region on a circle.

One may object that a complete characterization of a phase should also tell us the ordering of intervals within one color. As an example, in addition to the completely connected phase of three intervals $(X_1 X_2 X_3)$, there might conceivably exist an alternative phase $(X_1 X_3 X_2)$, which is also `completely connected.' This turns out not to be the case: the minimal configuration is always the one where successive geodesics connect intervals according to their spatial ordering. We prove this intuitive fact in Appendix~\ref{apporder}.

A final caveat about phases of entanglement entropy is that not all colorings are valid: in the spatial ordering, the colors of intervals must never alternate. As an example, the coloring $(X_1 X_3) (X_2 X_4)$ is forbidden, which is why we crossed it out from the list above. This rule reflects the requirement that the minimal surface which computes $S(I_l)$ be homologous to $I_l$. In a static situation, it means that the minimal surface may not cross itself.

Although in the preceding discussion we colored the connected components of $I_l$, in a pure state any such coloring uniquely colors $\bar{I_l}$, the complementary region. In particular, if $X_i$ and $X_j$ are consecutive intervals with the same color in $I_l$ then $Y_{i+1}$ and $Y_j$ have the same color in $\bar{I_l}$. The coloring of $I_l$ and of $\bar{I_l}$ contains equivalent information because specifying the phase of $S(I_l)$ also specifies the phase of $S(\bar{I_l})$.
\begin{figure}
        \centering
        \includegraphics[width=0.45\textwidth]{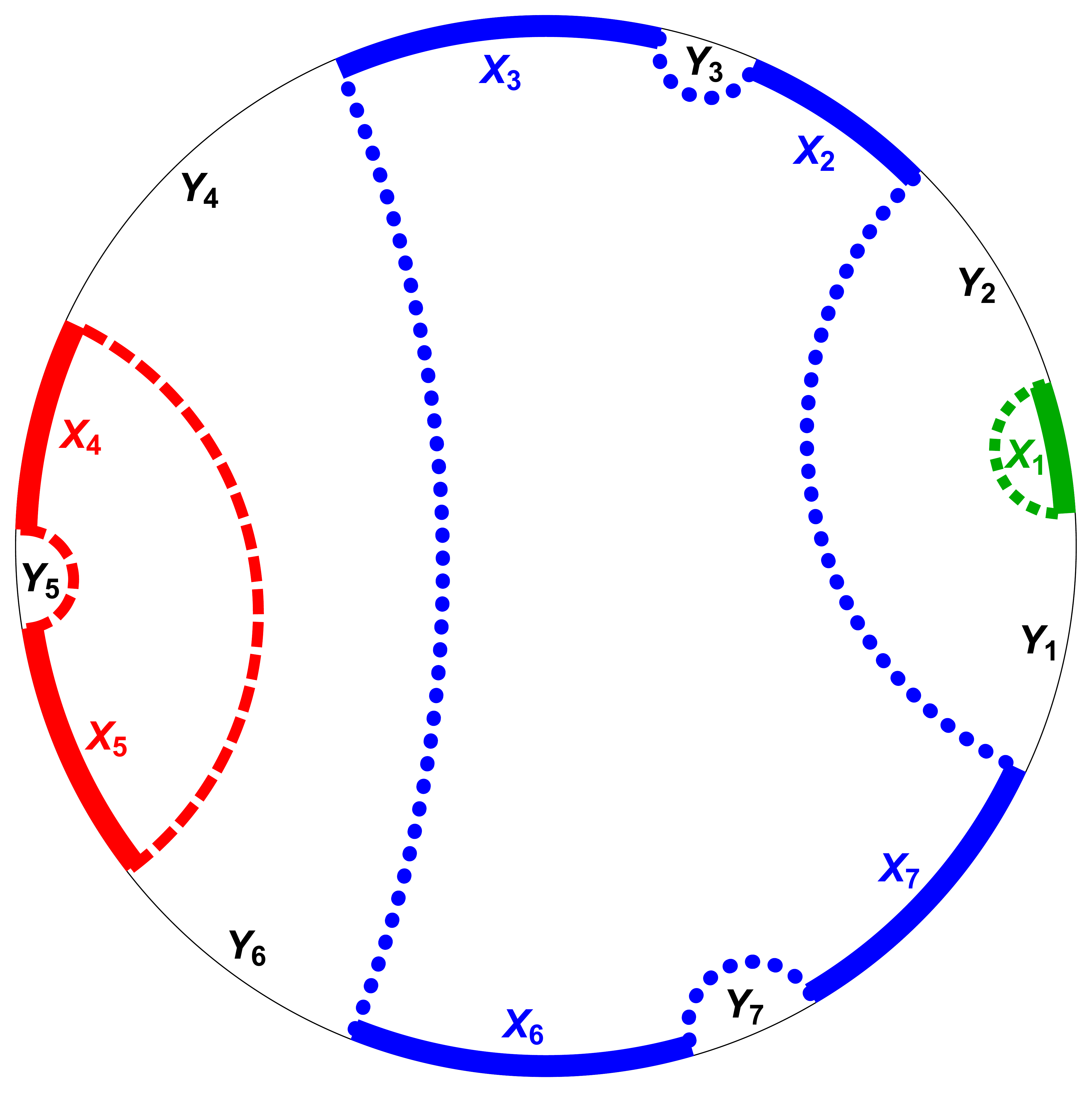}
        \hfill
        \raisebox{2.3cm}{\includegraphics[width=0.45\textwidth]{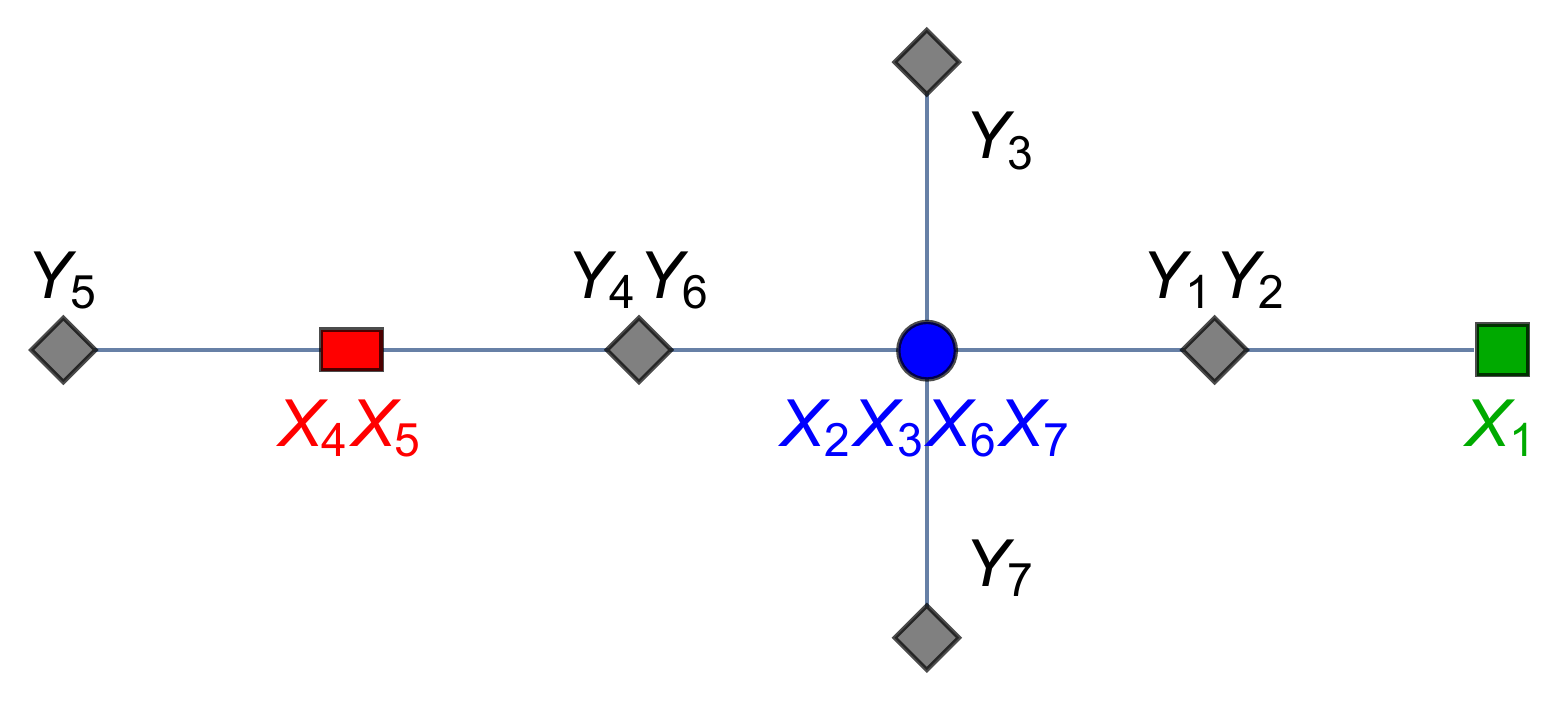}}
        \caption{The phase $(X_4 X_5) (X_2 X_3 X_6 X_7) (X_1)$ of a seven-interval region on a circle and its coloring tree (see text).}
        \label{fig:example}
\end{figure}

\subsection{Inequalities as adversarial games}

In every specific instance of the problem, verifying an inequality amounts to the following: given a coloring on the left, we must find one coloring on the right whose total geodesic length is no greater. Note that it is unnecessary to find the globally minimal coloring on the right; we only have to find one which is no greater than the left hand side. (If we succeed in finding one such coloring, the global minimum---if different---will be even smaller and the inequality will still hold.) Thus, we can recast the problem as an adversarial game: one player chooses a coloring for the left and her opponent's objective is to find an even smaller coloring for the right. To prove an inequality is to formulate a winning strategy for the second player.

The inequalities that we are mainly interested in were proven in \cite{hec} to hold in static configurations. Our goal is to show that in three bulk dimensions, every inequality proven in \cite{hec} also holds when the bulk space-time is time-dependent. To do so, we will reuse certain ingredients from the proof in \cite{hec} to formulate a winning strategy for the second player. Verifying that the resulting strategy guarantees a win will not involve the existence of a static bulk, but rely on purely boundary considerations by exploiting properties of kinematic space.

\subsection{Kinematic space}
\label{sec:ks}

In its most general form, kinematic space comprises arbitrary pairs of points from a CFT manifold \cite{intgeometry, stereoscopic}. Here we will draw the points from a spatial slice of the CFT on which the intervals $X_i$ and $Y_i$ live. In fact, even this notion of kinematic space is too detailed for our purposes: we will bin together points living in any one interval to form a discretized kinematic space whose coordinates are the $X_i$s and $Y_i$s themselves. Such a kinematic space can be represented as a symmetric matrix; it has the topology of $T^2/\mathbb{Z}_2$. We will denote our discretized kinematic space $\mathcal{K}$; see Fig.~\ref{fig:kspace}.

The coloring for the $l^{\rm th}$ term on the left produces a function $h_l$ on $\mathcal{K}$; we will call this function the overlap number. For an element $(Z_i, Z_j) \in \mathcal{K}$, $h_l(Z_i, Z_j)$ encodes how many colors separate intervals $Z_i$ from $Z_j$.
For instance, if $X_i, X_j \subset I_l$ have the same color, $h_l(X_i, X_j) = 0$ and, assuming $Y_j \neq \emptyset \neq Y_{j+1}$, $h_l(X_i, Y_j) = h_l(X_i, Y_{j+1}) = 1$. In order to formalize the definition of $h_l$, we have to introduce an auxiliary concept: the coloring tree (see the right panel of Fig.~\ref{fig:example}). 

\begin{figure}[t]
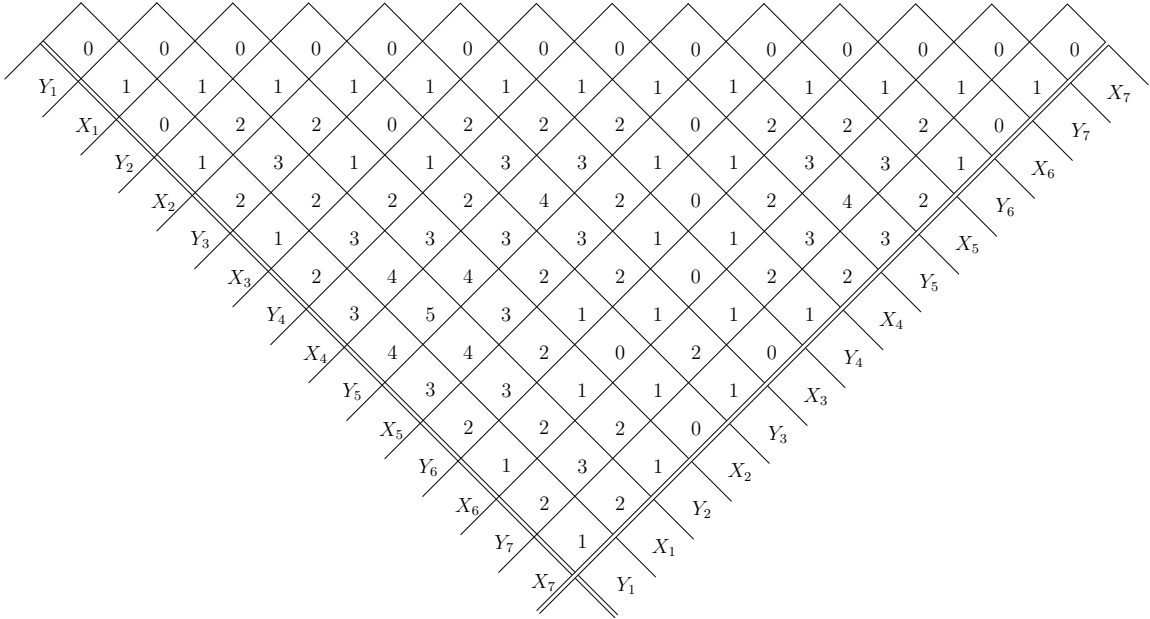

\vspace*{-.425\textwidth}
\renewcommand{\arraystretch}{1.795}
\rotatebox{45}{
\resizebox{0.7\textwidth}{!}{
\begin{tabular}{c||*{14}{C{.667cm}|}}
    \cline{1-2}
\rrb{$Y_1$} & \rrb{0} & \multicolumn{13}{r}{} \\
    \cline{1-3}
\rrb{$X_1$} & \rrb{1} & \rrb{0} & \multicolumn{12}{r}{} \\
     \cline{1-4}
\rrb{$Y_2$} & \rrb{0} & \rrb{1} & \rrb{0} & \multicolumn{11}{r}{}  \\
    \cline{1-5}
\rrb{$X_2$} & \rrb{1} & \rrb{2} & \rrb{1} & \rrb{0} & \multicolumn{10}{r}{} \\
        \cline{1-6}
\rrb{$Y_3$} & \rrb{2} & \rrb{3} & \rrb{2} & \rrb{1} & \rrb{0} & \multicolumn{9}{r}{} \\
    \cline{1-7}
\rrb{$X_3$} & \rrb{1} & \rrb{2} & \rrb{1} & \rrb{0} & \rrb{1} & 
    \rrb{0} & \multicolumn{8}{r}{} \\
    \cline{1-8}
\rrb{$Y_4$} & \rrb{2} & \rrb{3} & \rrb{2} & \rrb{1} & \rrb{2} & 
    \rrb{1} & \rrb{0} & \multicolumn{7}{r}{} \\
    \cline{1-9}
\rrb{$X_4$} & \rrb{3} & \rrb{4} & \rrb{3} & \rrb{2} & \rrb{3} & 
    \rrb{2} & \rrb{1} & \rrb{0} & \multicolumn{6}{r}{} \\
        \cline{1-10}
\rrb{$Y_5$} & \rrb{4} & \rrb{5} & \rrb{4} & \rrb{3} & \rrb{4} & 
    \rrb{3} & \rrb{2} & \rrb{1} & \rrb{0} & \multicolumn{5}{r}{} \\
    \cline{1-11}
\rrb{$X_5$} & \rrb{3} & \rrb{4} & \rrb{3} & \rrb{2} & \rrb{3} & 
    \rrb{2} & \rrb{1} & \rrb{0} & \rrb{1} & \rrb{0} & \multicolumn{4}{r}{} \\
    \cline{1-12}
\rrb{$Y_6$} & \rrb{2} & \rrb{3} & \rrb{2} & \rrb{1} & \rrb{2} & 
    \rrb{1} & \rrb{0} & \rrb{1} & \rrb{2} & \rrb{1} & 
    \rrb{0} & \multicolumn{3}{r}{} \\
    \cline{1-13}
\rrb{$X_6$} & \rrb{1} & \rrb{2} & \rrb{1} & \rrb{0} & \rrb{1} & 
    \rrb{0} & \rrb{1} & \rrb{2} & \rrb{3} & \rrb{2} & 
    \rrb{1} & \rrb{0} & \multicolumn{2}{r}{} \\
    \cline{1-14}
\rrb{$Y_7$} & \rrb{2} & \rrb{3} & \rrb{2} & \rrb{1} & \rrb{2} & 
    \rrb{1} & \rrb{2} & \rrb{3} & \rrb{4} & \rrb{3} & 
    \rrb{2} & \rrb{1} & \rrb{0} & \multicolumn{1}{r}{}  \\
    \cline{1-15}
\rrb{$X_7$} & \rrb{1} & \rrb{2} & \rrb{1} & \rrb{0} & \rrb{1} & 
\rrb{0} & \rrb{1} & \rrb{2} & \rrb{3} & \rrb{2} & 
\rrb{1} & \rrb{0} & \rrb{1} & \rrb{0} \\
    \hline
    \hline
& \rrb{$Y_1$} & \rrb{$X_1$} & \rrb{$Y_2$} & \rrb{$X_2$} & \rrb{$Y_3$}
 & \rrb{$X_3$} & \rrb{$Y_4$} & \rrb{$X_4$} & \rrb{$Y_5$} & \rrb{$X_5$}
  & \rrb{$Y_6$} & \rrb{$X_6$} & \rrb{$Y_7$}  & \rrb{$X_7$}
    \end{tabular}
}}
\caption{The kinematic space and overlap number $h(Z_i, Z_j)$ for the phase depicted in Fig.~\ref{fig:example}.}
\label{fig:kspace}
\end{figure}

To every color in the coloring of $I_l$ or of $\bar{I_l}$ we associate a vertex of a graph. Two vertices are connected if there exists a pair of contiguous intervals bearing their two colors. To see that a graph defined this way is necessarily a tree, consider one vertex, which corresponds to a color $c$ and intervals $(X_{i_1} X_{i_2} \ldots X_{i_k})$. Let us partition the remaining intervals into groups which fall in between the consecutive components of $c$:
\begin{align}
& \textrm{group 1:}~\{Y_{i_1+1}, X_{i_1+1}, \ldots, Y_{i_2}\} 
\nonumber \\
& \textrm{group 2:}~\{Y_{i_2+1}, X_{i_2+1}, \ldots, Y_{i_3}\} \\
& \ldots \nonumber
\end{align}
Each of these groups accounts for one edge adjacent to $c$ because $Y_{i_1+1}$ and $Y_{i_2}$ have the same color in $\bar{I_l}$. (When $Y_{i_1+1} = \emptyset$, substitute $X_{i_1+1}$ instead.) But in a valid coloring---one where colors do not alternate---no color can appear in more than one group. Thus, no two edges adjacent to $c$ can be part of a loop.

With recourse to the coloring tree, it is easy to define the overlap number. We set $h_l(Z_i, Z_j)$ to be the graph distance between the colors of $Z_i$ and $Z_j$ in the coloring tree. In static configurations, the overlap number has a more direct meaning: it takes a geodesic with endpoints in $Z_i$ and $Z_j$ and counts its intersections with the geodesics that compute $S(I_l)$. We stress, however, that the coloring tree and $h_l$ are purely boundary concepts and their definitions make no commitment to a static bulk. 
The overlap number for the phase depicted in Fig.~\ref{fig:example} is tabulated in Fig.~\ref{fig:kspace}.

We can combine the overlap numbers $h_l$ from different terms on the left hand side to form an aggregate quantity which we call the overlap function:
\begin{equation}
h_{LHS}(Z_i, Z_j) = \sum_{l=1}^L \alpha_l\, h_l(Z_i, Z_j). 
\label{defhlhs}
\end{equation}
This object counts the total overlap number of $(Z_i, Z_j)$ in a particular choice of colorings for all terms on the left hand side of (\ref{templateineq}). In our adversarial game, the right hand side player's choice of coloring can similarly be scored by an overlap function $h_{RHS}$. The importance of these functions is captured by the following two lemmas.

\subsection{Overlaps determine multiplicities of all geodesics}

\begin{lemma}
In every coloring of the left hand side, the overlap function $h_{LHS}$ completely characterizes the collection of geodesics that compute $\sum_l \alpha_l S(I_l)$.
Explicitly, let the minimal geodesic that connects the common endpoint of $Z_i$ and $Z_{i+1}$ to the common endpoint of $Z_j$ and $Z_{j+1}$ appear on the left hand side of (\ref{templateineq}) $k_{LHS}(i,j)$ times. This multiplicity of geodesics is a nonnegative integer determined by $h_{LHS}$ via:
\begin{equation}
k_{LHS}(i,j) = \frac{1}{2}\big(
  h_{LHS}(Z_i, Z_j) + h_{LHS}(Z_{i+1}, Z_{j+1})
- h_{LHS}(Z_i, Z_{j+1}) - h_{LHS}(Z_{i+1}, Z_j)
\big).
\label{htok}
\end{equation}
\end{lemma}

\begin{proof}
The geodesic under consideration contributes to a given term $S(I_l)$ if and only if on the coloring tree of $I_l$, $Z_i$ and $Z_{j+1}$ share one color, whereas $Z_{i+1}$ and $Z_j$ share a different color.  To diagnose this, consider the difference between
\begin{equation}\la{hdj}
h_l(Z_i, Z_j) - h_l(Z_{i+1}, Z_j) =
\begin{cases} 
\phantom{-}1 
& \begin{subarray}{l}\textrm{if $Z_i$ and $Z_{i+1}$ have different colors and} \\ \textrm{$Z_j$ is closer to $Z_{i+1}$ than to $Z_i$ on the coloring tree of $I_l$}\end{subarray} \\
-1 
& \begin{subarray}{l}\textrm{if $Z_i$ and $Z_{i+1}$ have different colors and} \\ \textrm{$Z_j$ is closer to $Z_{i}$ than to $Z_{i+1}$ on the coloring tree of $I_l$}\end{subarray} \\
\phantom{-}0 & \begin{subarray}{l}\textrm{otherwise (i.e., if $Z_i$ and $Z_{i+1}$ have the same color} \\ \textrm{on the coloring tree of $I_l$)}\end{subarray}
\end{cases}
\end{equation}
and
\begin{equation}\la{hdjp}
h_l(Z_i, Z_{j+1}) - h_l(Z_{i+1}, Z_{j+1}) =
\begin{cases} 
\phantom{-}1 
& \begin{subarray}{l}\textrm{if $Z_i$ and $Z_{i+1}$ have different colors and} \\ \textrm{$Z_{j+1}$ is closer to $Z_{i+1}$ than to $Z_i$ on the coloring tree of $I_l$}\end{subarray} \\
-1 
& \begin{subarray}{l}\textrm{if $Z_i$ and $Z_{i+1}$ have different colors and} \\ \textrm{$Z_{j+1}$ is closer to $Z_{i}$ than to $Z_{i+1}$ on the coloring tree of $I_l$}\end{subarray} \\
\phantom{-}0 & \begin{subarray}{l}\textrm{otherwise (i.e., if $Z_i$ and $Z_{i+1}$ have the same color} \\ \textrm{on the coloring tree of $I_l$)}\end{subarray}
\end{cases}
\end{equation}

First, in the case of
\be\la{gch}
h_l(Z_i, Z_j) - h_l(Z_{i+1}, Z_j) = 1,\qu
h_l(Z_i, Z_{j+1}) - h_l(Z_{i+1}, Z_{j+1}) = -1,
\ee
the geodesic under consideration must contribute to $S(I_l)$.  To see this, we use eqs.~\er{hdj} and \er{hdjp} to find that $Z_i$ and $Z_{i+1}$ have different colors, and $Z_j$ is closer to $Z_{i+1}$ than to $Z_i$, whereas $Z_{j+1}$ is closer to $Z_{i}$ than to $Z_{i+1}$.  Since by construction $Z_j$ and $Z_{j+1}$ share an edge on the coloring tree, going through $Z_j$, $Z_{i+1}$, $Z_i$, $Z_{j+1}$, and back to $Z_j$ would create a nontrivial loop, unless $Z_i$, $Z_{j+1}$ have the same color (corresponding to the same vertex) and $Z_{i+1}$, $Z_j$ have the same color.  This means that the geodesic under consideration contributes to $S(I_l)$.

Second, in all cases where \er{gch} is not satisfied, we must have
\be\la{hde}
h_l(Z_i, Z_j) - h_l(Z_{i+1}, Z_j) = h_l(Z_i, Z_{j+1}) - h_l(Z_{i+1}, Z_{j+1}).
\ee
To see this, we note that the left hand side of eq.~\er{hde} vanishes if and only if its right hand side vanishes.  Therefore, to show eq.~\er{hde} we only need to rule out the possibility of its left hand side being $-1$ and its right hand side being $1$.  An argument similar to the one in the previous paragraph (but with $Z_j$ and $Z_{j+1}$ exchanged) implies that $Z_i$, $Z_j$ share a color and $Z_{i+1}$, $Z_{j+1}$ share a different color, but this is an invalid coloring because the colors alternate.

Therefore, eq.~(\ref{htok}) holds because its right hand side simply collects contributions of a given geodesic from all terms in $\sum_l \alpha_l S(I_l)$.
\end{proof}

Of course, the same argument establishes that $h_{RHS}$ determines which geodesics (and with what multiplicities) comprise every phase of $\sum_r \beta_r S(J_r)$. In other words, $h_{LHS}$ and $h_{RHS}$ completely characterize the geodesics corresponding to the colorings of the $I_l$s and $J_r$s, which in turn specify phases of $\sum_l \alpha_l S(I_l)$ and $\sum_r \beta_r S(J_r)$. (We cannot, however, read off from $h_{LHS}$ or $h_{RHS}$ which phase is physically realized, i.e., which phase minimizes the total geodesic length.) Since $h_{LHS}$ and $h_{RHS}$ uniquely identify the geodesics that comprise each phase of $\sum_l \alpha_l S(I_l)$ and $\sum_r \beta_r S(J_r)$, it should be possible to recast the inequality in terms of overlap functions. This is accomplished by the next lemma.

\subsection{Inequality in terms of overlap functions}
\label{ineqoverlap}

\begin{lemma}
\label{ineqks}
Choose a coloring of each of the regions $I_l$ and $J_r$. If $h_{LHS} \geq h_{RHS}$ for all $(Z_i, Z_j) \in \mathcal{K}$ then inequality (\ref{templateineq}) holds {in this coloring}. 
\end{lemma}

\begin{proof}
Define
\begin{equation}
CMI(i,j) = 
S\left( \cup_{k = i}^{j-1} Z_k \right) +
S\left( \cup_{k = i+1}^{j} Z_k \right)-
S\left( \cup_{k = i}^{j} Z_k \right)-
S\left( \cup_{k = i+1}^{j-1} Z_k \right) \geq 0.
\label{defcmi}
\end{equation}
This quantity, a conditional mutual information, is non-negative by virtue of the strong subadditivity of entanglement entropy \cite{ssaref}, which the RT and HRT proposals are known to obey \cite{holstaticssa, maximin}. 
Each individual term in (\ref{defcmi}) is the entanglement entropy of a single interval (a union of contiguous intervals), so there are no phase ambiguities about which geodesics compute (\ref{defcmi}). For example, for the first term it is the geodesic connecting the common endpoint of $Z_{i-1}$ and $Z_{i}$ to the common endpoint of $Z_{j-1}$ and $Z_{j}$. 

In analogy to eq.~(\ref{defhlhs}), we may define and compute an overlap function for (\ref{defcmi}):
\begin{equation}
h_{CMI(i,j)}(Z_m, Z_n) = 2 \delta^{ij, mn}
\end{equation}
where $\delta^{ij,mn} = \delta^{im}\delta^{jn} + \delta^{in}\delta^{jm}$. The ordering of $i,j$ and $m,n$ does not matter because elements of $\mathcal{K}$ are un-ordered pairs and $CMI(i,j) = CMI(j,i)$ in pure states. The only pair of intervals which has a net overlap with $CMI(i,j)$ is $(Z_i, Z_j)$.

Given $h_{LHS}$ and $h_{RHS}$ obtained from some coloring of the terms of (\ref{templateineq}), consider
\begin{equation}
RHS' \equiv 
\sum_{r=1}^R \beta_r S(J_r) + 
\frac{1}{2} \sum_{i\neq j} 
\big( h_{LHS}(Z_i, Z_j) - h_{RHS}(Z_i, Z_j) \big)\, CMI(i,j).
\end{equation}
By construction we have $h_{LHS} = h_{RHS'}$ and therefore, by Lemma~1, both $LHS$ and $RHS'$ are computed by the same geodesics with the same multiplicities:\footnote{This is not affected by the presence of terms with negative coefficients in $RHS'$, as we can always move them to $LHS$.}
\begin{equation}
\sum_{l=1}^L \alpha_l S(I_l) = \sum_{r=1}^R \beta_r S(J_r) + 
\frac{1}{2} \sum_{i\neq j} 
\big( h_{LHS}(Z_i, Z_j) - h_{RHS}(Z_i, Z_j) \big)\, CMI(i,j).
\label{ineqisssa}
\end{equation}
If $h_{LHS} \geq h_{RHS}$ for all $(Z_i, Z_j) \in \mathcal{K}$, inequality~(\ref{templateineq}) follows (in this coloring).
\end{proof}

Observe that our strategy effectively reduces inequality (\ref{templateineq}) to instances of strong subadditivity. Our argument thus shows that in three bulk dimensions (at least in pure states of CFTs on connected manifolds) the difference between (\ref{templateineq}) and strong subadditivity is combinatorial in character. The combinatorics identify the applications of strong subadditivity, which make a specific instance of (\ref{templateineq}) manifest. This is reminiscent of the conclusions of \cite{intgeometry}, which argued that lengths of all space-like curves in three-dimensional holographic geometries could be calculated by adding up a number of conditional mutual information quantities or their close analogues \cite{entwinement}.

\begin{corollary} A winning strategy is an algorithm which, given a coloring of $I_l$ (terms on the left hand side), finds a coloring of $J_r$ (terms on the right hand side) such that $h_{LHS} \geq h_{RHS}$ for all $(Z_i, Z_j) \in \mathcal{K}$.
\end{corollary}

We will formulate such a strategy in Sec.~\ref{strategy}. Luckily, we will not have to start from scratch. Instead, we will take advantage of properties of \emph{contraction maps}, which allowed the authors of \cite{hec} to constrain (and determine for up to $n = 5$ named regions~\cite{Cuenca:2019uzx}) the static holographic entropy cone.

\se{Contraction map and the static proof}
\label{staticreview}

This section reviews relevant aspects of \cite{hec}. Here (only in this section) the setup is static: we consider a time-independent asymptotically AdS geometry that is holographically dual to some CFT state. The inequalities of interest now concern minimal surfaces, which are all taken from a common static slice of the bulk, as dictated by the RT proposal. The proof of \cite{hec} involves partitioning the static bulk slice into regions; this is the element we will remove in Sec.~\ref{strategy}.

We are interested in inequalities of the general form
\begin{equation}
\sum_{l=1}^L \alpha_l S(I_l) 
\geq 
\sum_{r=1}^R \beta_r S(J_r)\,,
\label{templateineq2}
\end{equation}
where $\alpha_l$ and $\beta_r$ are all positive. The sets $I_l$ and $J_r$ are unions of regions, which themselves can comprise multiple basic intervals $X_i$. For example, the monogamy of holographic mutual information
\begin{equation}
S(AB) + S(BC) + S(AC) \geq S(A) + S(B) + S(C) + S(ABC)
\label{monogamy0}
\end{equation}
has $I_1 = AB$, $I_2 = BC$, $I_3 = AC$ and $J_1 = A$, $J_2 = B$, $J_3 = C$, $J_4 = ABC$. (The regions $A, B, C$, which are assumed to be disjoint, can consist of an arbitrary number of basic intervals that we call $X_i$.) We have been calling the number of distinct regions involved in inequality~(\ref{templateineq2}) $n$, in keeping with the notation of \cite{hec}. Thus, for (\ref{monogamy0}) we have $n = 3$ ($A, B, C$), $l = 3$, and $r = 4$, while $N$, the number of \emph{basic} intervals, can be arbitrary.

The authors of \cite{hec} represented the composition of $I_l$ and $J_r$ in terms of regions $A, B, C, \ldots$ using $n+1$ pairs of \emph{occurrence vectors}. These are vectors in $L$- and $R$-dimensional spaces whose entries---zeroes and ones---encode whether or not the given region (say, $A$) is contained inside $I_l$ (respectively $J_r$):
\begin{equation}
\vec{x}_A := (A \subset I_l)^L_{l = 1} \in \{0,1\}^L
\qquad {\rm and} \qquad
\vec{y}_A := (A \subset J_r)^R_{r = 1} \in \{0,1\}^R.
\label{occurr}
\end{equation}
In addition to $n$ such vectors that correspond to $A, B, C, \ldots$ we also define $\vec{x}_O \equiv 0$ and $\vec{y}_O \equiv 0$, which represent the purifying region (the complement of the union of $A, B, C, \ldots$).

Our argument will make use of another concept defined in \cite{hec}: the contraction map. This is a function $f: \{0,1\}^L \to \{0,1\}^R$ which satisfies two properties:
\begin{align}
\sum_{l=1}^L \alpha_l |x_l - x'_l| \geq \sum_{r=1}^R \beta_r |f(\vec{x})_r - f(\vec{x\,}')_r|~~~ & \textrm{for all}~\vec{x}, \vec{x\,}' \in \{0,1\}^L, \label{defcontraction} \\
f(\vec{x}_S) = \vec{y}_S~~~ & \textrm{for}~ S = O, A, B, C, \ldots \quad (n+1~{\rm conditions}). \label{defic}
\end{align}
Eq.~(\ref{defcontraction}) requires that $f$ shorten distances between two vectors with respect to a norm, which is set by the coefficients in (\ref{templateineq2}). Eq.~(\ref{defic}) contains the `initial data' based on the composition of the regions $I_l$ and $J_r$.

\begin{table}
\centering
\begin{tabular}{l cc p{0.1cm} cc}
  \hline
  & \multicolumn{2}{c}{$\vec{x}$} & & \multicolumn{2}{c}{$\vec{y} = f(\vec{x})$} \\
  \cline{2-3} \cline{5-6}
      & ~$BC$~ & ~$AC$~ & & ~$ABC$~ & ~$C$~ \\
  $O$ & 0 & 0 & & 0 & 0  \\
  $A$    & 0 & 1 & & 1 & 0  \\
  $B$   & 1 & 0 & & 1 & 0 \\
  $C$ & 1 & 1 & & 1 & 1 \\
  \hline
\end{tabular}
\caption{The unique contraction map for the strong subadditivity inequality.}
\label{tab:ssa-contraction}
\end{table}

One of the main results of \cite{hec} is that the existence of a contraction map guarantees the correctness of inequality (\ref{templateineq2}), assuming that the bulk geometry is static. The authors of \cite{hec} then found explicit contraction maps for multiple inequalities, which for $n\leq 5$ fully characterize the static holographic entropy cone \cite{Cuenca:2019uzx}. For example, the contraction map that proves the strong subadditivity inequality $S(BC) + S(AC) \geq S(ABC) + S(C)$ is shown in Table~\ref{tab:ssa-contraction}.

In Sec.~\ref{strategy} we explain how to use a contraction map to define a winning strategy in our adversarial game. This demonstrates that every inequality proven in \cite{hec} also holds in a time-dependent geometry in three bulk dimensions. Importantly, we will not need to find any new contraction maps. Instead, we explain how to devise a winning strategy \emph{given} that a contraction map exists.

\subsection{How a contraction map works in the static case}
Before returning to time-dependent settings, it is illustrative to review how contraction maps prove inequalities when the bulk is static. They do this by partitioning the static slice of the bulk.

Observe that the minimal surface that computes $S(I_l)$ divides the bulk static slice into two parts. One of them, which \cite{hec} denotes $W_l$, sits between the minimal surface and $I_l$ on the asymptotic boundary. (The requirement that the minimal surface which computes $S(I_l)$ be homologous to $I_l$ boils down to demanding that $W_l$ exist and be uniquely defined.) The other part of the bulk slice is its complement, $\overline{W_l}$. We draw two examples of $W_l$-type regions in Fig.~\ref{wlsbw}.

\begin{figure}
        \centering
        \includegraphics[width=0.31\textwidth]{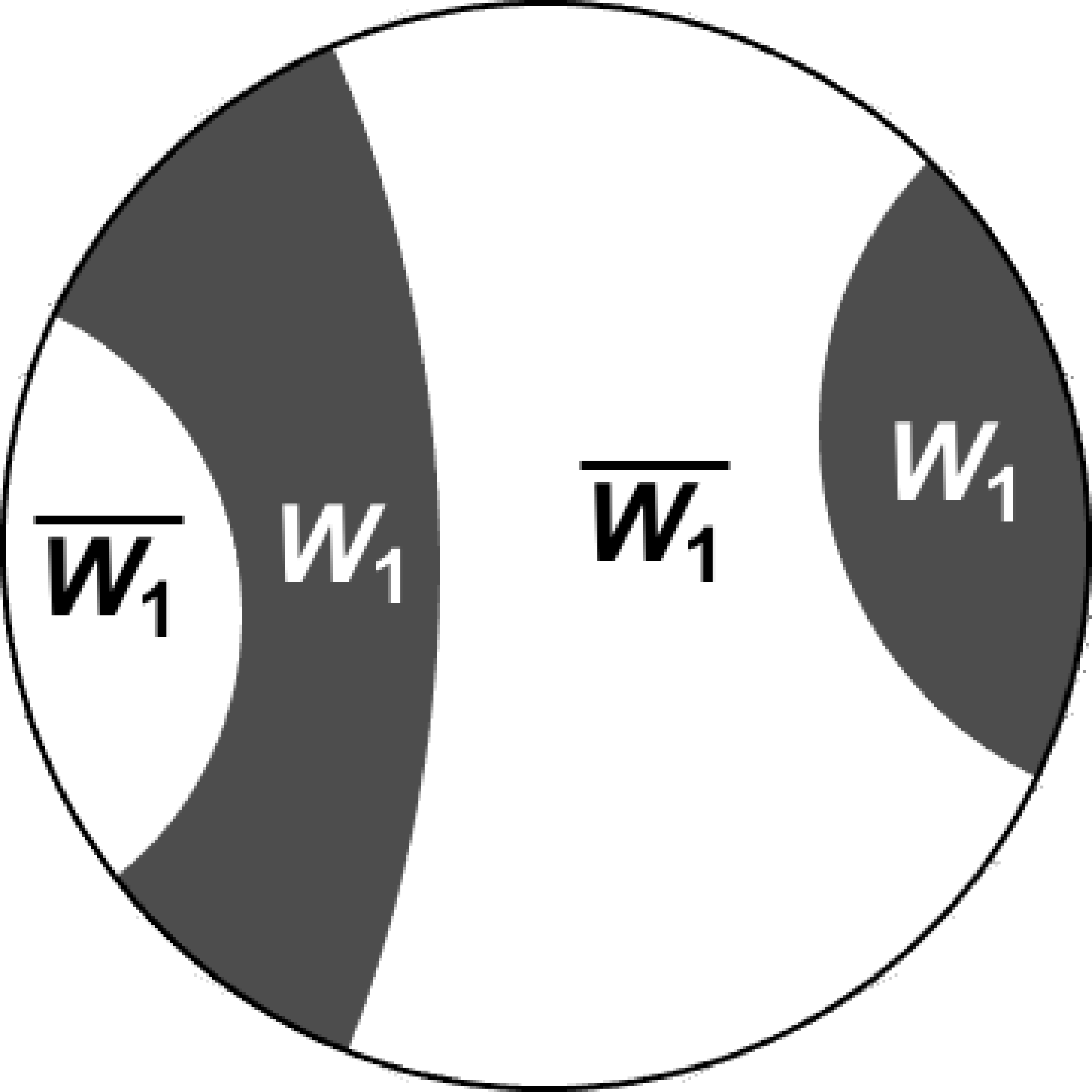}
     \hfill
        \includegraphics[width=0.31\textwidth]{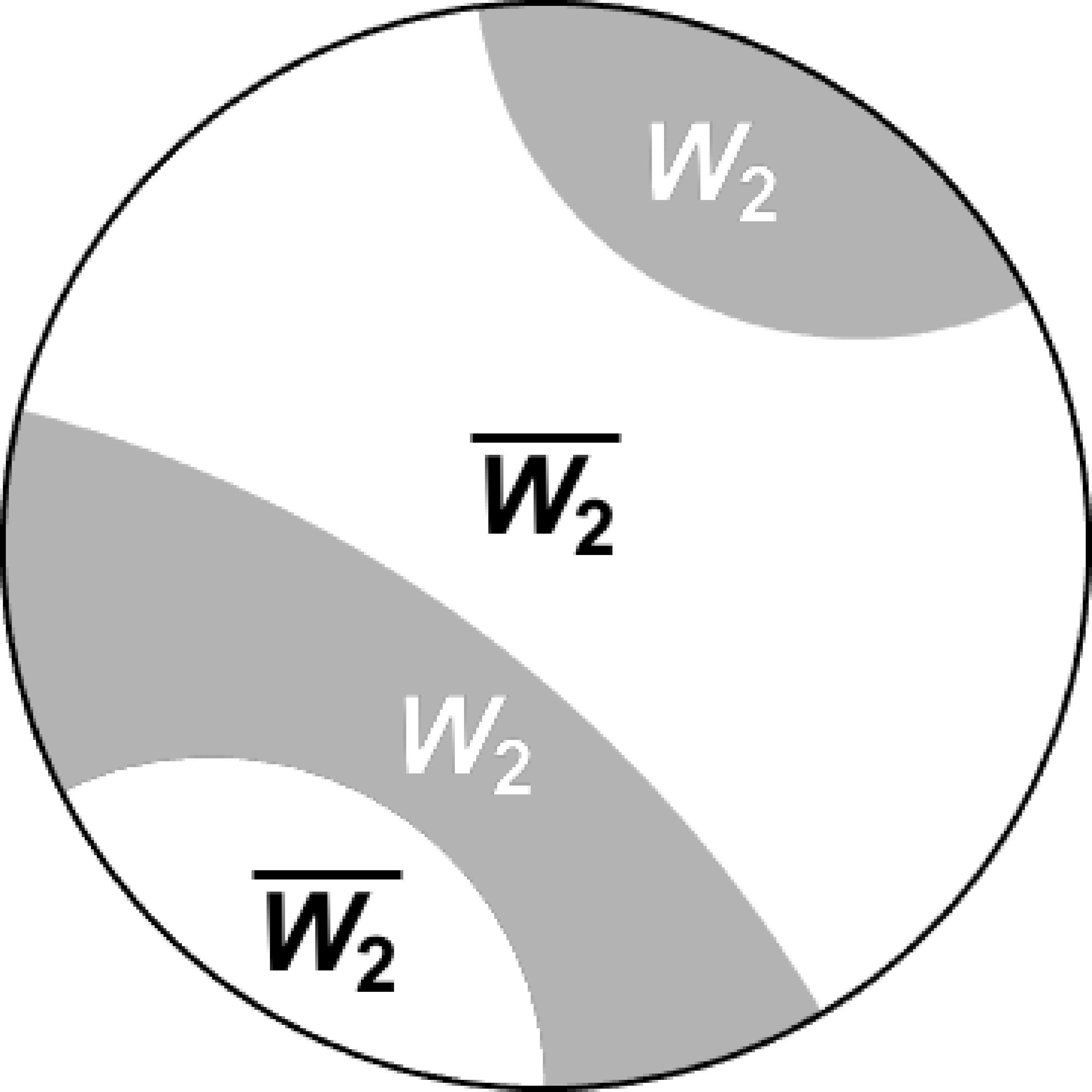}
        \hfill
        \includegraphics[width=0.31\textwidth]{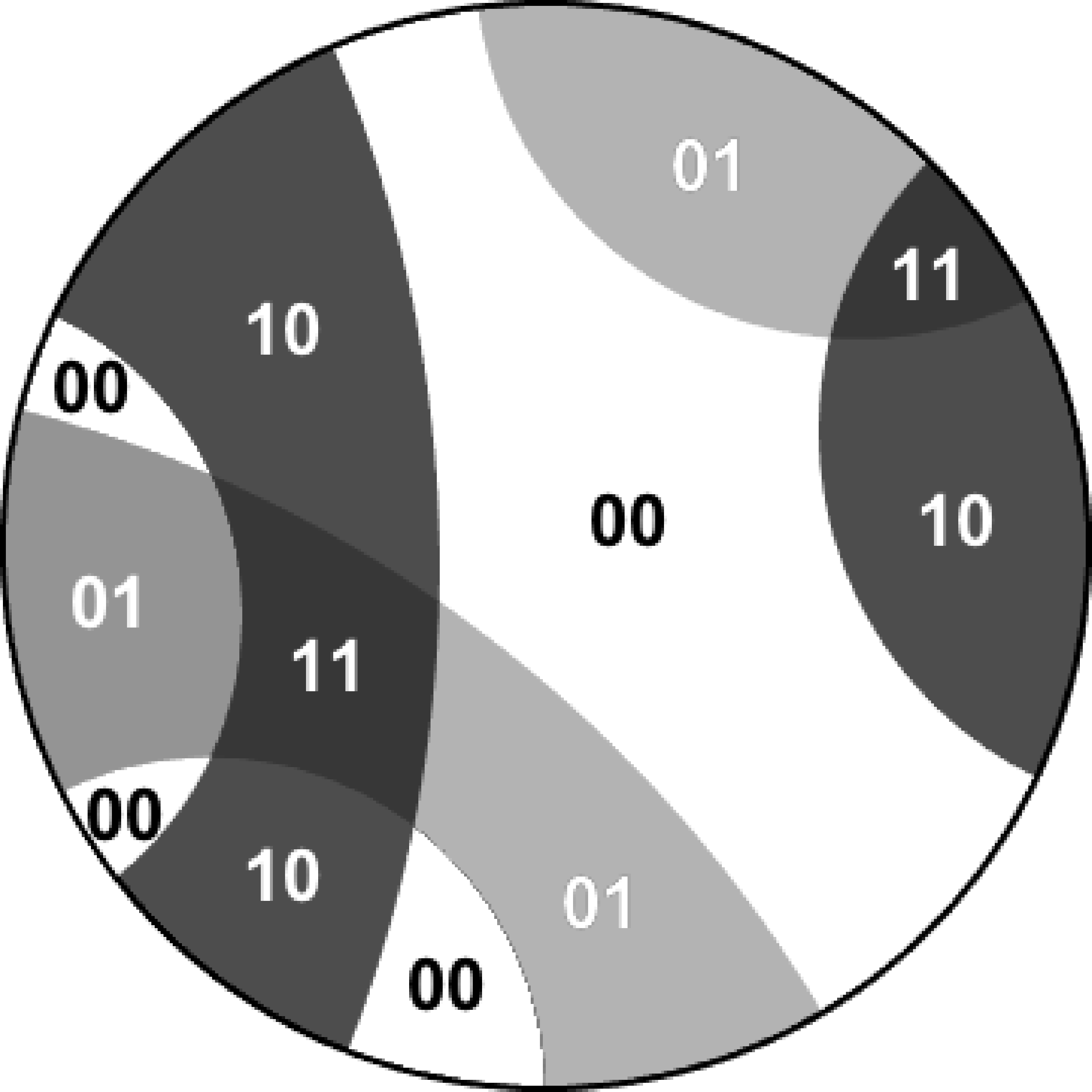} 
   \caption{Two examples of regions $W_l$ (left and center) and the regions $W(\vec{x})$ labeled with the components of $\vec{x}$ (right).}
   \label{wlsbw}
\end{figure}

By taking intersections of $W_l$s and $\overline{W_l}$s, we can partition the bulk static slice into $2^L$ regions; see the right panel of Fig.~\ref{wlsbw}. In easily constructible examples, many of them may be empty and many others---disconnected. In a general case, however, we have defined $2^L$ disjoint regions whose union is the static bulk slice. We will associate these subregions of the bulk slice with vectors in $\{0, 1\}^L$ and call them $W(\vec{x})$:
\be
W(\vec{x}) := \bigcap_{l \, | \, x_l = 1} W_l \, \cap \bigcap_{l \, | \, x_l = 0} \overline{W_l}
\qquad \textrm{for}~~\vec{x} \in \{0, 1\}^L.
\ee
An observation that will be useful momentarily is that among the $2^L$ $W({\vec{x}})$s, exactly $n+1$ reach the asymptotic boundary; these are the $\vec{x}_S$s which we saw in eqs.~(\ref{occurr}) and (\ref{defic}).

In a given phase (coloring) of the left hand side, the authors of \cite{hec} give the following prescription for selecting a `winning' right hand side. Let
\begin{equation}
U_r\,  = \bigcup_{\vec{x}\, | \, f(\vec{x})_r = 1} W({\vec{x}}).
\label{defuj}
\end{equation}
In other words, take all vectors $\vec{x} \in \{0,1\}^L$ whose image under the contraction map has one in its $r^{\rm th}$ component and form a union over all of their $W({\vec{x}})$s. 
Every $U_r$ is now a region on the static bulk slice whose boundary has two parts: one on the asymptotic boundary and the other in the bulk. On the asymptotic boundary, the boundary of $U_r$ is $J_r$. This is a consequence of the $n+1$ `initial data' in eq.~(\ref{defic}), which guarantee that amongst the $n+1$ $W({\vec{x}})$s that reach the boundary, only those which border the components of $J_r$ are included in union (\ref{defuj}). 

The bulk part of the boundary of $U_r$ is some bulk surface; call it $\mathcal{A}_r$. By construction, $\mathcal{A}_r$ is homologous to $J_r$ and therefore its area is greater than or equal to $S(J_r)$ (because the latter is obtained by minimizing over all bulk surfaces homologous to $J_r$, including $\mathcal{A}_r$). Thus, inequality~(\ref{templateineq2}) will follow if one can show
\begin{equation}
\sum_{l=1}^L \alpha_l S(I_l) \geq \sum_{r=1}^R \beta_r \mathcal{A}_r.
\label{staticsandwich}
\end{equation}
This is guaranteed by the contraction condition~(\ref{defcontraction}). 

To see this, think of $\mathcal{A}_r$ as a collection of component surfaces $\mathcal{A}_{\vec{x},\vec{x}'}$, each of which separates one $W({\vec{x}})$ (such that $f(\vec{x})_r = 1$) from some other $W({\vec{x}'})$ (such that $f(\vec{x}')_r = 0$). A component surface $\mathcal{A}_{\vec{x},\vec{x}'}$ also appears on the left hand side of (\ref{staticsandwich}) and contributes with multiplicity $\sum_{l=1}^L \alpha_l |x_l - x'_l|$ because every non-zero entry of $\vec{x} - \vec{x}'$ identifies one $S(I_l)$ to which the common border of $W({\vec{x}})$ and $W({\vec{x}'})$ belongs. With these observations, inequality~(\ref{staticsandwich}) becomes
\begin{equation}
\frac{1}{2} \sum_{\vec{x}, \vec{x}'} \mathcal{A}_{\vec{x},\vec{x}'} \sum_{l=1}^L \alpha_l |x_l - x'_l| 
\geq 
\frac{1}{2} \sum_{\vec{x}, \vec{x}'} \mathcal{A}_{\vec{x},\vec{x}'} \sum_{r=1}^R \beta_r |f(\vec{x})_r - f(\vec{x\,}')_r|\,,
\label{staticsandwich2}
\end{equation}
which is guaranteed by the contraction property. The factors of $1/2$ in (\ref{staticsandwich2}) correct a double-counting under the exchange $\vec{x} \leftrightarrow \vec{x\,}'$.

This construction appears to heavily depend on the existence of a static bulk. In particular, the atomic objects that the proof manipulates are the components $\mathcal{A}_{\vec{x},\vec{x}'}$ of minimal surfaces, which cannot even be defined in a non-static setup. As we presently explain, the proof in fact carries over to the dynamical context if, instead of components of minimal areas, we manipulate the overlap numbers in kinematic space. 


\se{The strategy}
\label{strategy}

We now return to the general, time-dependent setup. As a first step, even though the bulk geometry contains no privileged spatial slice, we will invent one. In other words, we choose an auxiliary spatial geometry with the topology of a disk whose boundary is identified with the CFT time slice that contains the regions $X_i$ and $Y_j$.

At this point, it is probably important to appease the reader and emphasize that the role of the chosen geometry is purely auxiliary. Our proof will not draw from the choice of geometry any quantitative inputs such as surface areas; instead, the auxiliary geometry will only help us make certain arbitrary discrete choices in our coloring game. Moreover, any choice of auxiliary geometry will work for conducting the proof. In this sense, the freedom in choosing the auxiliary geometry parameterizes the flexibility and/or redundancy in proving inequalities~(\ref{templateineq}) using contraction maps. 

\sse{Auxiliary geometry}
\la{auxiliary}

\begin{figure}
        \centering
        \includegraphics[width=0.31\textwidth]{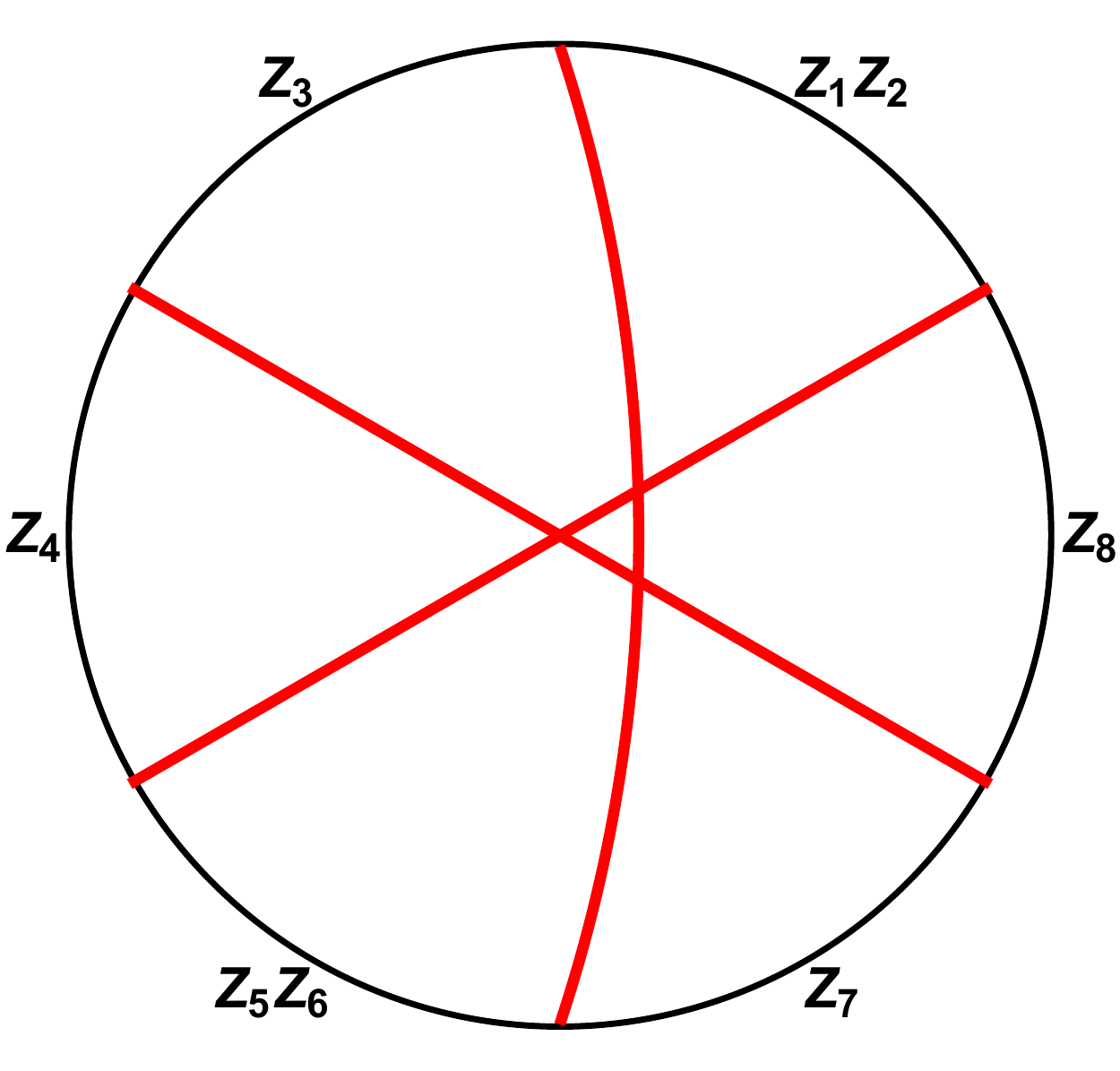}
     \hfill
        \includegraphics[width=0.31\textwidth]{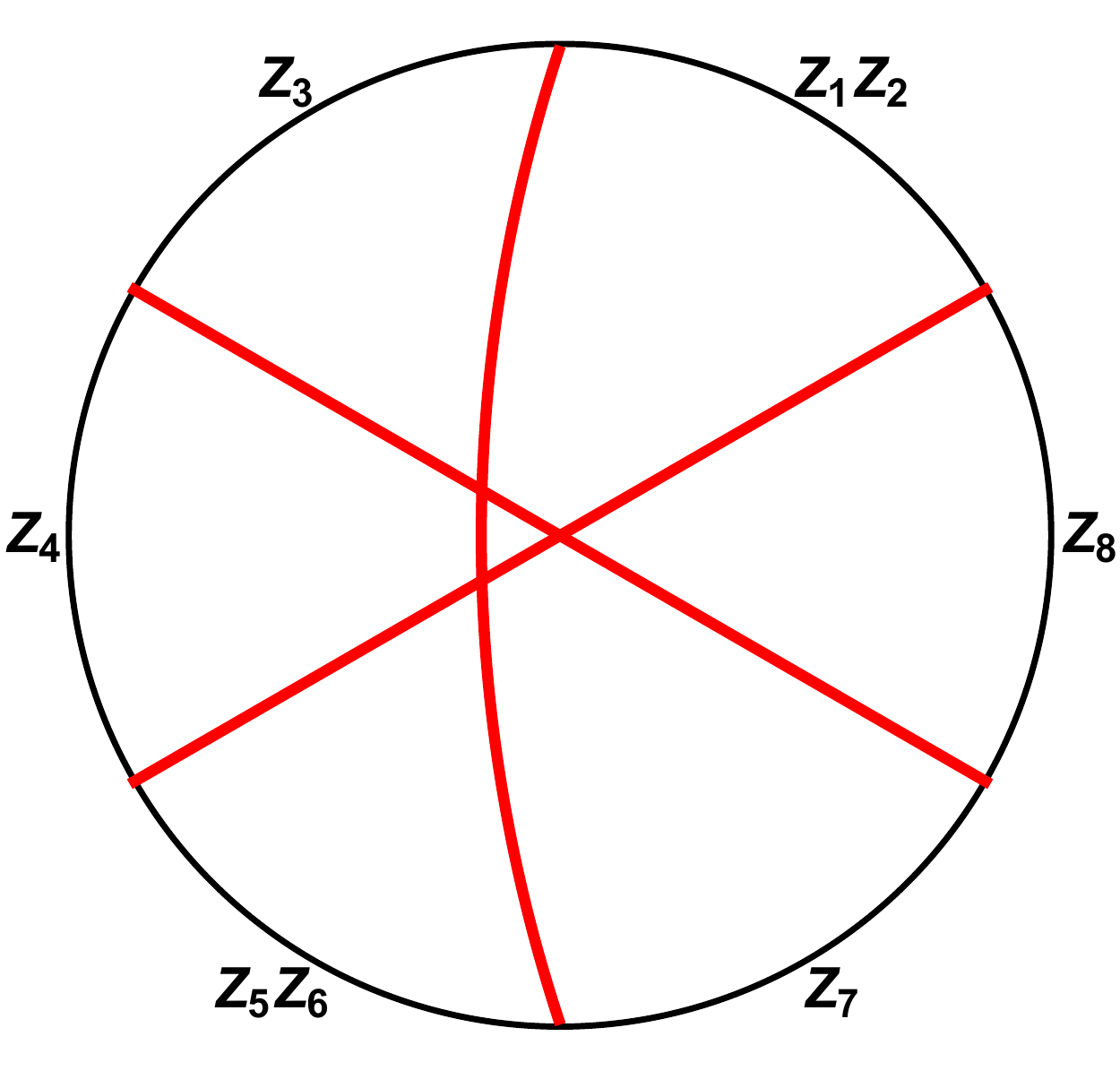}
        \hfill
        \includegraphics[width=0.31\textwidth]{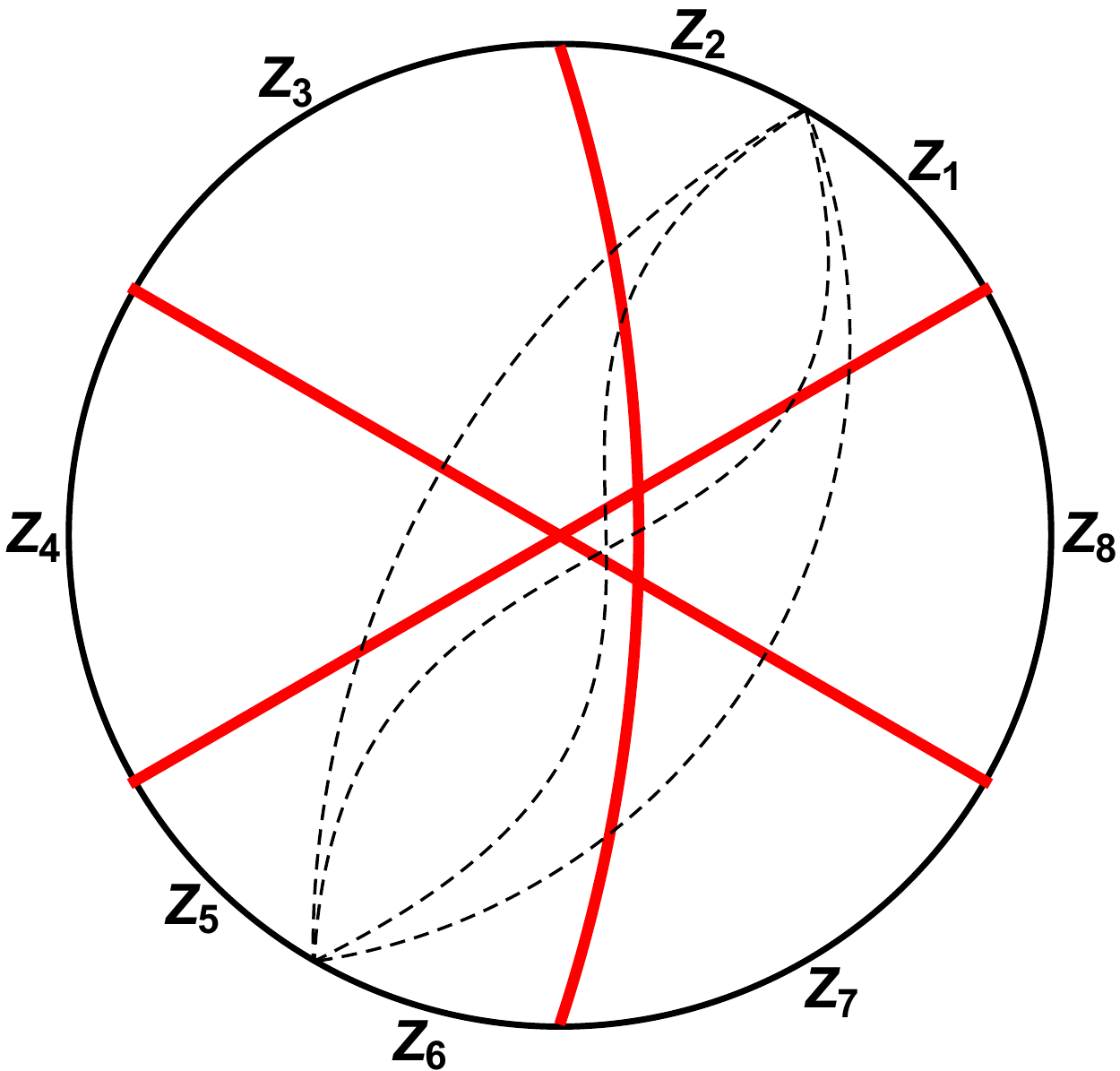} 
   \caption{Left and center: the two auxiliary geometries for three mutually intersecting geodesics. Right: Eight auxiliary geometries can be chosen for four mutually intersecting geodesics; here we draw the four possible ways to pass the fourth geodesic given the three-geodesic choice on the left.}
   \label{auxs}
\end{figure}

The important point in choosing an auxiliary geometry is to decide which geodesics intersect on which side of other geodesics. To illustrate this in more detail, consider three geodesics $\gamma_{1,2,3}$ such that none of their subtended intervals contains another. In a static geometry, this would mean that the geodesics intersect pairwise. Now, even though we do not have a physical geometry, we will make a choice and \emph{declare} geodesics $\gamma_1$ and $\gamma_2$ to `intersect' either on one side or on the other side of $\gamma_3$. We have two possible choices, which are shown in the left two panels in Fig.~\ref{auxs}. 

When choosing an auxiliary geometry, we make such a choice for every triple of geodesics whose subtended intervals partially overlap pairwise. In order not to have to repeat this clunky phrase, from now on we will simply call them `mutually intersecting geodesics,'\footnote{In kinematic space, this means a triple of geodesics which are pairwise \emph{spacelike separated} \cite{intgeometry}.} with the understanding that their `intersection' happens in the auxiliary space, not in the physical bulk geometry (where they will in general bypass one another, separated in time). These choices are not all independent; for example, given four mutually intersecting geodesics $\gamma_{1,2,3,4}$, we would na{\"\i}vely have $2^{\binom{4}{3}}=16$ possible configurations but only 8 of them can be realized; see the right panel in Fig.~\ref{auxs}. 

Proceeding in analogy to the static case, we now define $W_l$ and $\overline{W_l}$ in the auxiliary geometry. A given coloring of $I_l$ selects a set of geodesics, which subtend components of $I_l$. Those geodesics and $I_l$ form the boundary of a submanifold in the auxiliary geometry, which is $W_l$; $\overline{W_l}$ is its complement. Thus far, we are using the same definitions as in the static case, except that they are applied to the auxiliary geometry. 

\begin{figure}
        \centering
        \includegraphics[width=0.32\textwidth]{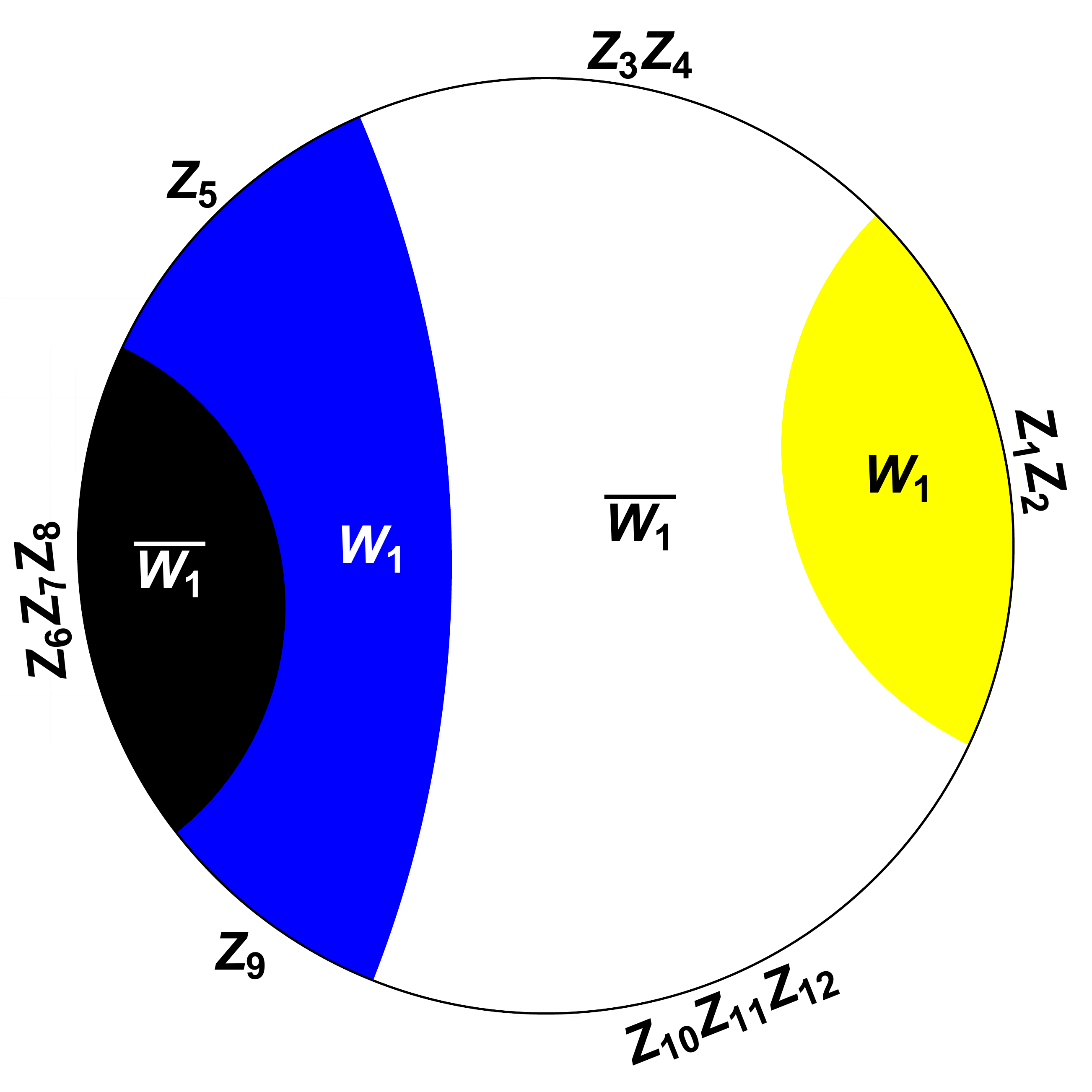}
     \hfill
        \includegraphics[width=0.32\textwidth]{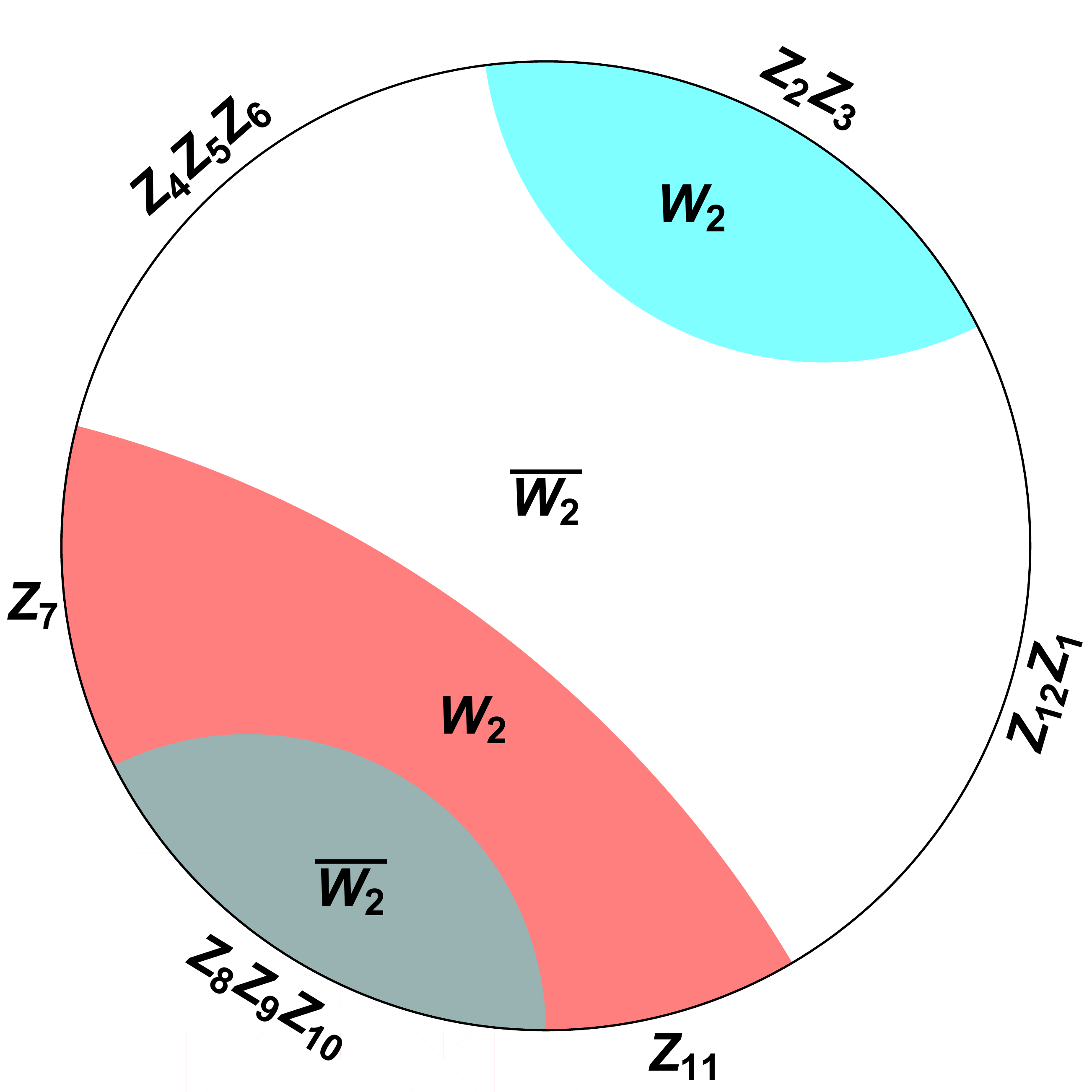}
        \hfill
        \includegraphics[width=0.32\textwidth]{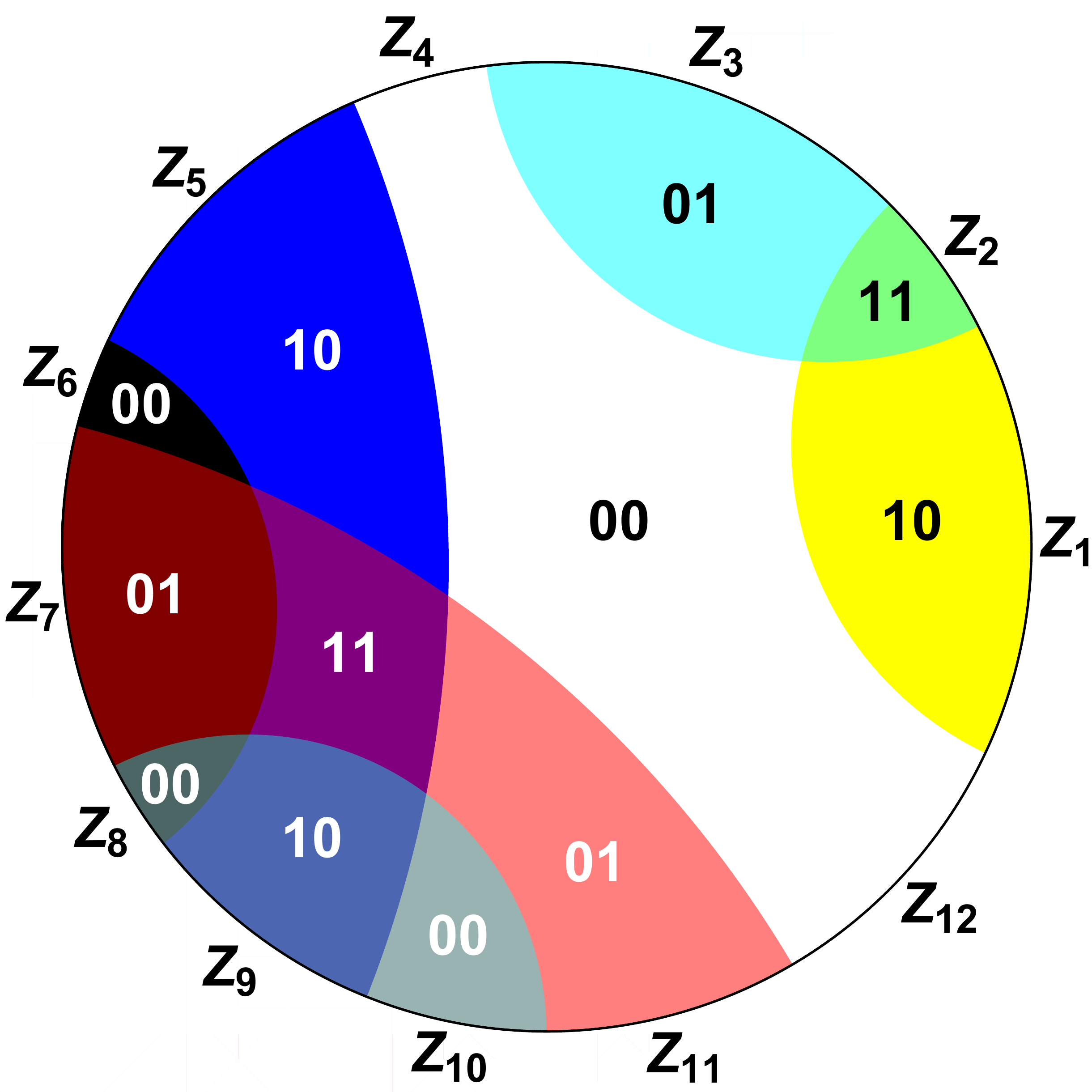} 
   \caption{Left and center panels: the division of regions $W_l$ and $\overline{W_l}$ from Fig.~\ref{wlsbw} into connected components. We have $W_1 = W_1^{\rm blue} \sqcup W_1^{\rm yellow}$, $\overline{W_1} = \overline{W_1}^{\rm white} \sqcup \overline{W_1}^{\rm black}$, \mbox{$W_2 = W_2^{\rm cyan} \sqcup W_2^{\rm red}$}, and $\overline{W_2} = \overline{W_2}^{\rm white} \sqcup \overline{W_2}^{\rm gray}$. Right panel: Regions $W(\vec{x})^{\vec{c}}$, i.e., the various intersections of $W_l^{c_l}$ and $\overline{W_l}^{c_l}$s. For example, the purple interior region is $W((1,1))^{({\rm blue},\, {\rm red})} = W_1^{\rm blue} \cap W_2^{\rm red}$ while the cyan region adjacent to interval $Z_3$ is $W((0,1))^{({\rm white},\, {\rm cyan})} = \overline{W_1}^{\rm white} \cap W_2^{\rm cyan}$. In this example, $I_1 = Z_1 Z_2 Z_5 Z_9$ and $I_2 = Z_2 Z_3 Z_7 Z_{11}$.}
   \label{wlscolor}
\end{figure}

In fact, we will need slightly finer objects than the $W_l$s: their connected components. We shall denote the connected components of $W_l$ with $W_l^c$, where $c$ labels the color\footnote{The superscript $c$ has nothing to do with complements (which are denoted by bars).} in the given coloring (phase) of $I_l$. For example, a phase of $I_l$ depicted in Fig.~\ref{fig:example} (now we interpret the figure as the auxiliary geometry) has three connected components: $W_l^{\rm red} = W_l^{(X_4 X_5)}$, $W_l^{\rm blue} = W_l^{(X_2 X_5 X_6 X_7)}$, and $W_l^{\rm green} = W_l^{(X_1)}$. Another example of colored $W_l^c$ regions, which is based on Fig.~\ref{wlsbw}, is displayed in the left and center panels of Fig.~\ref{wlscolor}. Because we have assumed to work in a pure state on a circle or line, the coloring of $I_l$ automatically colors $\overline{I_l}$, so we will have analogously defined regions $\overline{W_l}^c$. 

Proceeding in parallel with the static case, we now take all possible intersections of the regions $W_l^c$ and $\overline{W_l}^c$. The resulting subregions, which form a partition of the auxiliary geometry, are defined as
\begin{equation}
W(\vec{x})^{\vec{c}} = \bigcap_{l \, | \, x_l = 1} W_l^{c_l} \, \cap \bigcap_{l \, | \, x_l = 0} \overline{W_l}^{c_l}
\qquad \textrm{for}~~\vec{x} \in \{0, 1\}^L.
\label{wxcolor}
\end{equation}
Here $\vec{c}$ is a vector of colors of length $L$: if $x_l = 1$, $c_l$ labels a color in $I_l$; otherwise it labels a color in $\bar{I_l}$. An example of regions $W(\vec{x})^{\vec{c}}$ is shown in the right panel of Fig.~\ref{wlscolor}. 
Among them, there are $2N$ special ones, which are contiguous to the boundary of the auxiliary geometry. They correspond to the basic intervals $X_i$ and $Y_j$ on the CFT slice. We will call them $W(\vec{x}(X_i))^{\vec{c}(X_i)}$ (and similarly for $Y_j$).\footnote{Some of these $2N$ regions could be identical.  For example, $W(\vec{x}(X_i))^{\vec{c}(X_i)} = W(\vec{x}(X_j))^{\vec{c}(X_j)}$ if for all $l$, $X_i$ and $X_j$ share a color in the coloring of $I_l$ or of $\bar{I_l}$.\la{fid}} As an illustration, we tabulate them for the example of Fig.~\ref{wlscolor} in Table~\ref{tab:wxxcx}.

\begin{table}
\centering
\begin{tabular}{l c p{0.1cm} c p{0.1cm} l}
  \hline
$Z_i$ & $\vec{x}(Z_i)$ & & $\vec{c}(Z_i)$ & & \quad\qquad\qquad$W(\vec{x}(Z_i))^{\vec{c}(Z_i)}$ \\
  \hline
$Z_1$ & $(1,0)$ & & $({\rm yellow},\,{\rm white})$ & & $W((1,0))^{({\rm yellow},\, {\rm white})} = 
{W_1}^{\rm yellow} \cap \overline{W_2}^{\rm white}$ \phantom{\raisebox{.1cm}{\Big|}}\\
$Z_2$ & $(1,1)$ & & $({\rm yellow},\,{\rm cyan})$ & & $W((1,1))^{({\rm yellow},\, {\rm cyan})} = 
{W_1}^{\rm yellow} \cap W_2^{\rm cyan}$ \\
$Z_3$ & $(0,1)$ & & $({\rm white},\,{\rm cyan})$ & & $W((0,1))^{({\rm white},\, {\rm cyan})} = \overline{W_1}^{\rm white} \cap {W_2}^{\rm cyan}$ \\
$Z_4$ & $(0,0)$ & & $({\rm white},\,{\rm white})$ & & $W((0,0))^{({\rm white},\, {\rm white})} = \overline{W_1}^{\rm white} \cap \overline{W_2}^{\rm white}$ \\
$Z_5$ & $(1,0)$ & & $({\rm blue},\,{\rm white})$ & & $W((1,0))^{({\rm blue},\, {\rm white})} = 
{W_1}^{\rm blue} \cap \overline{W_2}^{\rm white}$ \\
$Z_6$ & $(0,0)$ & & $({\rm black},\,{\rm white})$ & & $W((0,0))^{({\rm black},\, {\rm white})} = \overline{W_1}^{\rm black} \cap \overline{W_2}^{\rm white}$ \\
$Z_7$ & $(0,1)$ & & $({\rm black},\,{\rm red})$ & & $W((0,1))^{({\rm black},\, {\rm red})} = \overline{W_1}^{\rm black} \cap W_2^{\rm red}$ \\
$Z_8$ & $(0,0)$ & & $({\rm black},\,{\rm gray})$ & & $W((0,0))^{({\rm black},\, {\rm gray})} = \overline{W_1}^{\rm black} \cap \overline{W_2}^{\rm gray}$ \\
$Z_9$ & $(1,0)$ & & $({\rm blue},\,{\rm gray})$ & & $W((1,0))^{({\rm blue},\, {\rm gray})} = 
{W_1}^{\rm blue} \cap \overline{W_2}^{\rm gray}$ \\
$Z_{10}$ & $(0,0)$ & & $({\rm white},\,{\rm gray})$ & & $W((0,0))^{({\rm white},\, {\rm gray})} = \overline{W_1}^{\rm white} \cap \overline{W_2}^{\rm gray}$ \\
$Z_{11}$ & $(0,1)$ & & $({\rm white},\,{\rm red})$ & & $W((0,1))^{({\rm white},\, {\rm red})} = \overline{W_1}^{\rm white} \cap W_2^{\rm red}$ \\
$Z_{12}$ & $(0,0)$ & & $({\rm white},\,{\rm white})$ & & $W((0,0))^{({\rm white},\, {\rm white})} = \overline{W_1}^{\rm white} \cap \overline{W_2}^{\rm white}$ \\
  \hline
\end{tabular}
\caption{The regions $W(\vec{x}(Z_i))^{\vec{c}(Z_i)}$, which are contiguous to the boundary, tabulated for the example from Fig.~\ref{wlscolor}. The instance of $Z_4$ and $Z_{12}$ illustrates that they do not have to be distinct.}
\label{tab:wxxcx}
\end{table}

Taking the union of $W(\vec{x})^{\vec{c}}$ over all colors gives us a region, which is directly analogous to $W(\vec{x})$ in  the static case:
\begin{equation}\la{defwx}
W(\vec{x}) = \bigcup_{\vec{c}}\, W(\vec{x})^{\vec{c}} = 
\bigcap_{l \, | \, x_l = 1} \! W_l \, \cap \bigcap_{l \, | \, x_l = 0} \!\overline{W_l}
\qquad \textrm{for}~~\vec{x} \in \{0, 1\}^L.
\end{equation}
We have seen an illustration of regions $W(\vec{x})$ in Fig.~\ref{wlsbw}. Among the $W(\vec{x})$s, $n+1$ are special in that they are contiguous to the boundary of the auxiliary disk. These special ones correspond to the $n$ regions $A, B, \ldots$ that populate the inequality in question, plus an $(n+1)^{\rm st}$ one that corresponds to the purifying region.

As one final object, we define for every $1 \leq r \leq R$
\begin{equation}
U_r = \bigcup_{\vec{x}\, | \, f(\vec{x})_r = 1} W({\vec{x}}).
\label{defujaux}
\end{equation}
This is the same definition as in the static case, except that it is applied to the auxiliary geometry. 

\sse{The winning coloring}
\label{winningstr}

The following strategy for coloring every $J_r$ on the right hand side of inequality~\er{templateineq} wins the game:  two basic intervals $X_j, X_k \subset J_r$ have the same color in $J_r$ if and only if their corresponding regions $W(\vec{x})^{\vec{c}}$ belong to the same connected component of $U_r$. 

Let us explain this rule in more detail. As we remarked below eq.~(\ref{wxcolor}), a unique $W(\vec{x})^{\vec{c}}$ meets the asymptotic boundary of the auxiliary geometry at $X_j \subset J_r$; we call it $W(\vec{x}(X_j))^{\vec{c}(X_j)}$. Of course, interval $X_k \subset J_r$ is also contiguous to a unique $W(\vec{x}(X_k))^{\vec{c}(X_k)}$. The two subregions are contained in $U_r$ defined in (\ref{defujaux}) because by construction $U_r$ meets the boundary at $J_r$ and $X_j, X_k \subset J_r$. However, $W(\vec{x}(X_j))^{\vec{c}(X_j)}$ and $W(\vec{x}(X_k))^{\vec{c}(X_k)}$ may or may not belong to the same connected component of $U_r$. 
In our coloring game, the winning strategy is to assign the same color to $X_j$ and $X_k$ in $J_r$ if and only if $W(\vec{x}(X_j))^{\vec{c}(X_j)}$ and $W(\vec{x}(X_k))^{\vec{c}(X_k)}$ belong to the same connected component of $U_r$.\footnote{If $W(\vec{x}(X_j))^{\vec{c}(X_j)}$ and $W(\vec{x}(X_k))^{\vec{c}(X_k)}$ are identical (as pointed out in footnote~\re{fid}), they automatically belong to the same connected component of $U_r$.} 

This prescription is well-defined because belonging to the same connected component of $U_r$ is an equivalence relation. 

\sse{Proof: why the strategy wins}
\label{winproof}

The strategy defined above picks a coloring of every $J_r$ on the right hand side of inequality~(\ref{templateineq}). As such, it defines an overlap number $h_r(Z_j, Z_k)$ on kinematic space. 
Our task is to prove that $h_{LHS}(Z_j,Z_k)\ge h_{RHS}(Z_j,Z_k)$, i.e.,
\begin{equation}
\sum_{l=1}^L \alpha_l\, h_l(Z_j, Z_k) \geq \sum_{r=1}^R \beta_r\, h_r(Z_j, Z_k)
\qquad \textrm{for all}~(Z_j, Z_k) \in \mathcal{K}\,.
\label{whattoprove}
\end{equation}

Consider a sequence $P$ of regions $W(\vec{x}(i))^{\vec{c}(i)}$ ($i = 1, 2, \ldots$) in our partition of the auxiliary geometry. The sequence is such that every two consecutive regions are neighbors across one geodesic in the auxiliary geometry. In this way, $P$ is a path in the auxiliary geometry, which involves a sequence of jumps across geodesics---from  $W(\vec{x}(i))^{\vec{c}(i)}$ to $W(\vec{x}(i+1))^{\vec{c}(i+1)}$.

For any region $V$ that is a union of some of the $W(\vec{x})^{\vec{c}}$s in the auxiliary geometry, let $h(P, V)$ count the number of times $P$ enters or exits $V$. For example, if $P$ consists only of $W(\vec{x})^{\vec{c}}$s that are part of $V$, $h(P,V) = 0$. If $P$ consists only of $W(\vec{x})^{\vec{c}}$s that are outside $V$, $h(P,V)$ also vanishes.

Our proof relies on the following two lemmas:
\begin{lemma}
\label{colorregionlemma}
Let $V$ be the union of some of the regions $W(\vec{x})^{\vec{c}}$ in the auxiliary geometry. Say that it meets the asymptotic boundary at $J_V$, which is a union of basic intervals $Z_{i_1} \cup Z_{i_2} \cup \ldots$.
We define a coloring of $J_V$ by the rule that two intervals $Z_j$ and $Z_k$ are in the same color if and only if $W(\vec{x}(Z_j))^{\vec{c}(Z_j)}$ and $W(\vec{x}(Z_k))^{\vec{c}(Z_k)}$ are in the same connected component of $V$. This coloring defines an overlap number $h_V(Z_j, Z_k)$ on kinematic space. 

Let $P$ be a path in the auxiliary space that begins at $W(\vec{x}(Z_j))^{\vec{c}(Z_j)}$ and ends at $W(\vec{x}(Z_k))^{\vec{c}(Z_k)}$. Then the following holds
\begin{equation}
h(P, V) \geq h_V(Z_j, Z_k)\qquad \textrm{for all}~(Z_j, Z_k) \in \mathcal{K}\,.
\label{colorfromregion}
\end{equation}
\end{lemma}

\begin{proof}
Consider the coloring tree of $J_V$, viz.\ Sec.~\ref{sec:ks} and Fig.~\ref{fig:example}. The overlap number $h_V(Z_j, Z_k)$ counts the edges of the coloring tree that must be crossed in going from the color of $Z_j$ to the color of $Z_k$.  Every edge in the coloring tree corresponds to a geodesic in the auxiliary geometry, which separates a color in $J_V$ from a color in $\bar{J}_V$. Thus, $h_V(Z_j, Z_k)$ identifies and counts certain geodesics that must be crossed at least once by any path from $W(\vec{x}(Z_j))^{\vec{c}(Z_j)}$ to $W(\vec{x}(Z_k))^{\vec{c}(Z_k)}$ in our auxiliary space. This is precisely the content of inequality~(\ref{colorfromregion}). The inequality is saturated if and only if the path $P$ makes no unnecessary geodesic crossings, i.e., if it never backtracks in the coloring tree.
\end{proof}

Note that Lemma~\ref{colorregionlemma} applies directly to functions $h_l(Z_j, Z_k)$ and $h_r(Z_j, Z_k)$ because by construction
\begin{equation}
h_l(Z_j, Z_k) = h_{W_l}(Z_j, Z_k) \quad {\rm and} \quad h_r(Z_j, Z_k) = h_{U_r}(Z_j, Z_k) \,.
\end{equation}
Therefore, for any path $P$ from $W(\vec{x}(Z_j))^{\vec{c}(Z_j)}$ to $W(\vec{x}(Z_k))^{\vec{c}(Z_k)}$ in the auxiliary geometry, we have:
\begin{eqnarray}
h(P, W_l) & \geq & h_l(Z_j, Z_k)\label{hpwlineq} \\
h(P, U_r) & \geq & h_r(Z_j, Z_k).\label{hpurineq}
\end{eqnarray}
The second lemma below states that for a given $(Z_j, Z_k) \in \mathcal{K}$, there is a path $P$ such that inequality (\ref{hpwlineq}) is simultaneously saturated for all $l$.

\begin{lemma}
\label{shortpaths}
There exists a path $P$ in the auxiliary space that begins at $W(\vec{x}(Z_j))^{\vec{c}(Z_j)}$, ends at $W(\vec{x}(Z_k))^{\vec{c}(Z_k)}$, and satisfies
\begin{equation}
h(P, W_l) = h_l(Z_j, Z_k) \qquad \textrm{for all}~~1 \leq l \leq L. \label{phl}
\end{equation}
\end{lemma}  

\begin{proof}
Saturating (\ref{hpwlineq}) means that $P$ cannot backtrack in the coloring tree of $I_l$. Any type of backtracking means crossing and re-crossing a geodesic in the auxiliary geometry.

Draw a geodesic $\gamma$ in the auxiliary geometry that starts at some point in the interior of interval $Z_j$ and ends at some point in $Z_k$. By virtue of being a geodesic, $\gamma$ cannot cross any other geodesic more than once. (Here we rely on the assumption that the geometry represents a pure state, so that two minimal geodesics can cross at most once.) Let $P$ be the sequence of regions $W(\vec{x}(i))^{\vec{c}(i)}$ traversed by $\gamma$. By construction, $P$ is a path that never backtracks in the coloring tree of any $I_l$, so (\ref{phl}) holds.
\end{proof}

\begin{proof}[Proof that the strategy of Sec.~\ref{winningstr} wins.]
For every $(Z_j, Z_k) \in \mathcal{K}$, take a path $P$ stipulated in Lemma~\ref{shortpaths}. Using (\ref{hpurineq}) and (\ref{phl}), to prove~\er{whattoprove} it suffices to show
\begin{equation}
\sum_{l=1}^L \a_l\, h(P,W_l) \ge \sum_{r=1}^R \b_r\, h(P,U_r).
\label{pathintersects}
\end{equation}
In fact, instead of considering the entire path $P$ in one go, we can split $P$ into individual jumps from $W(\vec{x}(i))^{\vec{c}(i)}$ to $W(\vec{x}(i+1))^{\vec{c}(i+1)}$. Let $P_i$ be a single jump in $P$, i.e., a one-step path from $W(\vec{x}(i))^{\vec{c}(i)}$ to $W(\vec{x}(i+1))^{\vec{c}(i+1)}$. Since
\begin{equation}
h(P, W_l) = \sum_i h(P_i, W_l) \qquad {\rm and} \qquad h(P, U_r) = \sum_i h(P_i, U_r)\,,
\label{pstepwise}
\end{equation}
it is enough to prove
\begin{equation}
\sum_{l=1}^L \a_l\, h(P_i,W_l) \ge \sum_{r=1}^R \b_r\, h(P_i,U_r).
\label{finalsimp}
\end{equation}
for every step $P_i: W(\vec{x}(i))^{\vec{c}(i)} \to W(\vec{x}(i+1))^{\vec{c}(i+1)}$ in path $P$. 

But the contraction property~(\ref{defcontraction}) guarantees that inequality~(\ref{finalsimp}) holds whenever we jump from any $W(\vec{x})^{\vec{c}}$ to a neighboring $W(\vec{x}')^{\vec{c}'}$. To see this, we recognize from the definition of $W_l$ that $h(P_i, W_l)$ is simply
\begin{equation}
h(P_i, W_l) = |x(i+1)_l - x(i)_l|.
\end{equation}
In other words, the jump $W(\vec{x}(i))^{\vec{c}(i)} \to W(\vec{x}(i+1))^{\vec{c}(i+1)}$ enters or leaves $W_l$ if and only if the vectors $\vec{x}(i)$ and $\vec{x}(i+1)$ differ in their $l^{\rm th}$ entries.
Similarly,
\begin{equation}
h(P_i, U_r) = |f(\vec{x}(i+1))_r - f(\vec{x}(i))_r|\,,
\end{equation}
that is, in jumping from $W(\vec{x}(i))^{\vec{c}(i)}$ to $W(\vec{x}(i+1))^{\vec{c}(i+1)}$ we enter or leave $U_r$ if and only if the vectors $f(\vec{x}(i))$ and $f(\vec{x}(i+1))$ differ in their $r^{\rm th}$ coordinate.

In this way, inequality~(\ref{finalsimp}) becomes
\begin{equation}
\sum_{l=1}^L \a_l\, |x(i+1)_l - x(i)_l| \ge \sum_{r=1}^R \b_r\, |f(\vec{x}(i+1))_r - f(\vec{x}(i))_r|\,,
\end{equation}
which follows from the contraction property~(\ref{defcontraction}).
\end{proof}

Therefore, we have proved inequality~\er{templateineq} in general, time-dependent settings using the existence of a contraction map.


\se{Example}
\label{examples}

We illustrate the winning strategy defined in the previous section using the following inequality:
\begin{equation}
S(AB) + S(BC) + S(AC) \geq S(AB) + S(ABC) + S(C).
\label{exineq}
\end{equation}
This is simply the strong subadditivity inequality $S(BC) + S(AC) \geq S(ABC) + S(C)$ plus the tautology $S(AB) \geq S(AB)$. 
This inequality showcases several noteworthy features of our proof, but it is sufficiently transparent so the reader will not be distracted by adventitious features of the example. We will use the contraction map $f(\vec{x})$ displayed in Table~\ref{tab:example}. 

\begin{table}
\centering
\begin{tabular}{l ccc p{0.1cm} ccc}
  \hline
  & \multicolumn{3}{c}{$\vec{x}$} & & \multicolumn{3}{c}{$\vec{y} = f(\vec{x})$} \\
  \cline{2-4} \cline{6-8}
      & ~$AB$~ & ~$BC$~ & ~$AC$~ & & ~$AB$~ & ~$ABC$~ & ~$C$~ \\
  $O$~~ & 0 & 0 & 0 & & 0 & 0 & 0 \\
      & 0 & 0 & 1 & & 0 & 1 & 0 \\
      & 0 & 1 & 0 & & 0 & 1 & 0 \\
  $C$ & 0 & 1 & 1 & & 0 & 1 & 1 \\
      & 1 & 0 & 0 & & \framebox{0} & \framebox{1} & 0 \\
  $A$  & 1 & 0 & 1 & & 1 & 1 & 0 \\
  $B$  & 1 & 1 & 0 & & 1 & 1 & 0 \\
      & 1 & 1 & 1 & & 1 & 1 & 1 \\
  \hline
\end{tabular}
\caption{A non-trivial contraction map for a linear combination of strong subadditivity and $S(AB) \geq S(AB)$, inequality~(\ref{exineq}). The $AB$ column in the image of $f(\vec{x})$ differs from the $AB$ column on the domain side of $f$ in the boxed entry. The $ABC$-column of $f(\vec{x})$ differs from the $ABC$-column in Table~\ref{tab:ssa-contraction} (the strong subadditivity contraction map) in the boxed entry.}
\label{tab:example}
\end{table}

\paragraph{Comment 1:} Because we did not cancel the tautology $S(AB) \geq S(AB)$ out of inequality~(\ref{exineq}), region $AB$ has a column on both the domain and the range side of $f(\vec{x})$. The simplest contraction map for (\ref{exineq}) would have set both $AB$-columns equal to one another ($f(\vec{x})_{AB}=x_{AB}$) and copied the entries of the strong subadditivity contraction map (Table~\ref{tab:ssa-contraction}) to the columns $f(\vec{x})_{ABC}$ and $f(\vec{x})_{C}$. The resulting strategy in the coloring game would be to color $ABC$ and $C$ in the same way as in proving strong subadditivity, and to color $AB$ as it was given on the left hand side. Such a strategy effectively cancels $S(AB)$ out of both sides of inequality~(\ref{exineq}). 

\begin{figure}
        \centering
        \includegraphics[width=0.39\textwidth]{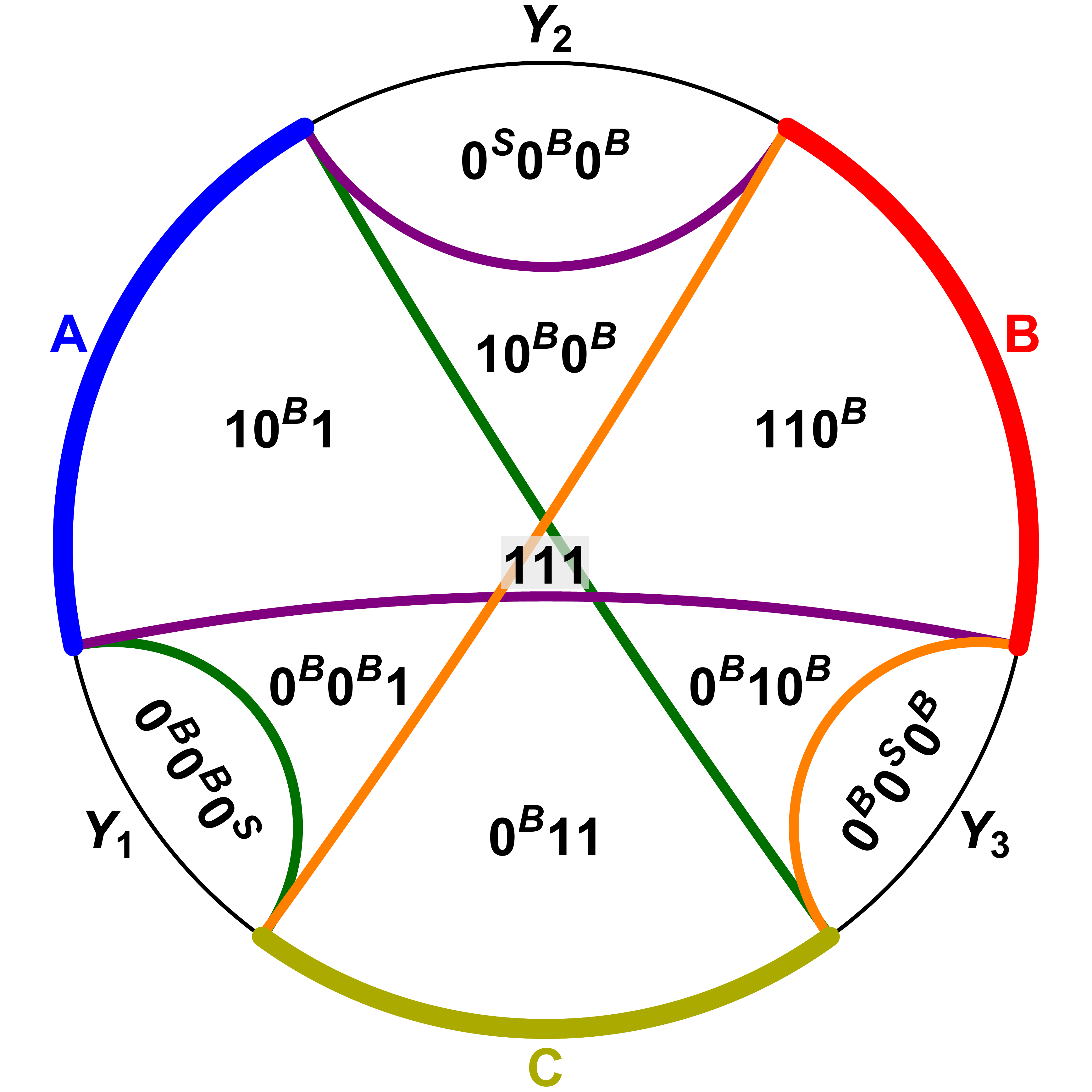}
        \hfill
        \includegraphics[width=0.39\textwidth]{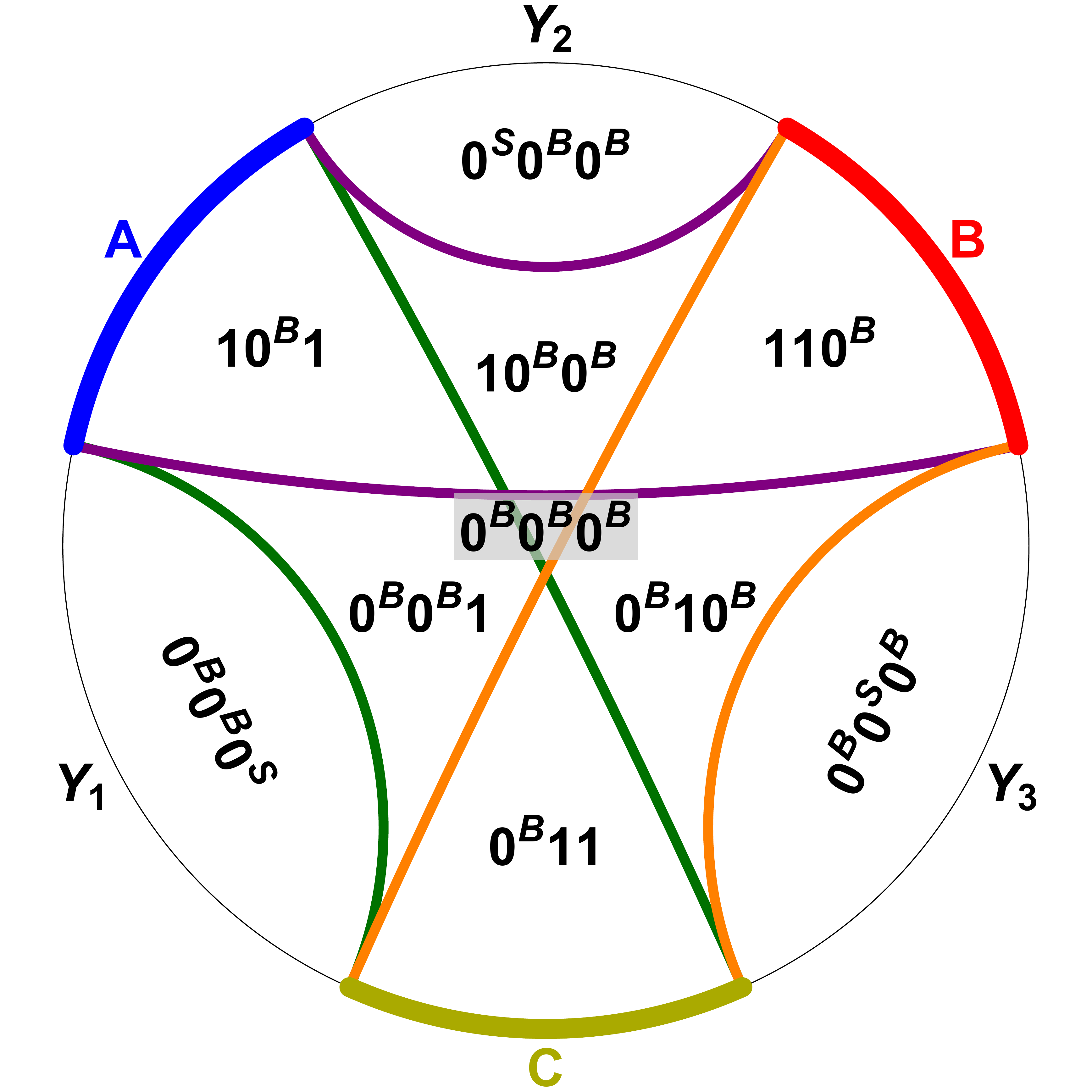}
        \caption{Auxiliary geometries for inequality~(\ref{exineq}): Geometry I (left) and Geometry II (right). The boundaries of region $W_{AB}$ are marked in purple, the boundaries of $W_{AC}$ are marked in green, and the boundaries of $W_{BC}$ are marked in orange. We label the regions $W(\vec{x})^{\vec{c}}$ with their $\vec{x}$-vectors, e.g., $101$ stands for $W_{AB} \cap \overline{W_{BC}} \cap W_{AC}$. Because each $\overline{W_{I_l}}$ is in the disconnected phase, every 0-entry in $\vec{x}$ carries a superscript that labels the color component of $\overline{W_{I_l}}$: B for the `bigger' component and S for the `smaller' component (a single $Y$-interval). No color label is given for 1-entries in $\vec{x}$ because each $W_{I_l}$ is connected (and therefore has a single color). 
        }
        \label{auxgeoms}
\end{figure}

\begin{figure}[h]
        \centering
        \includegraphics[width=0.39\textwidth]{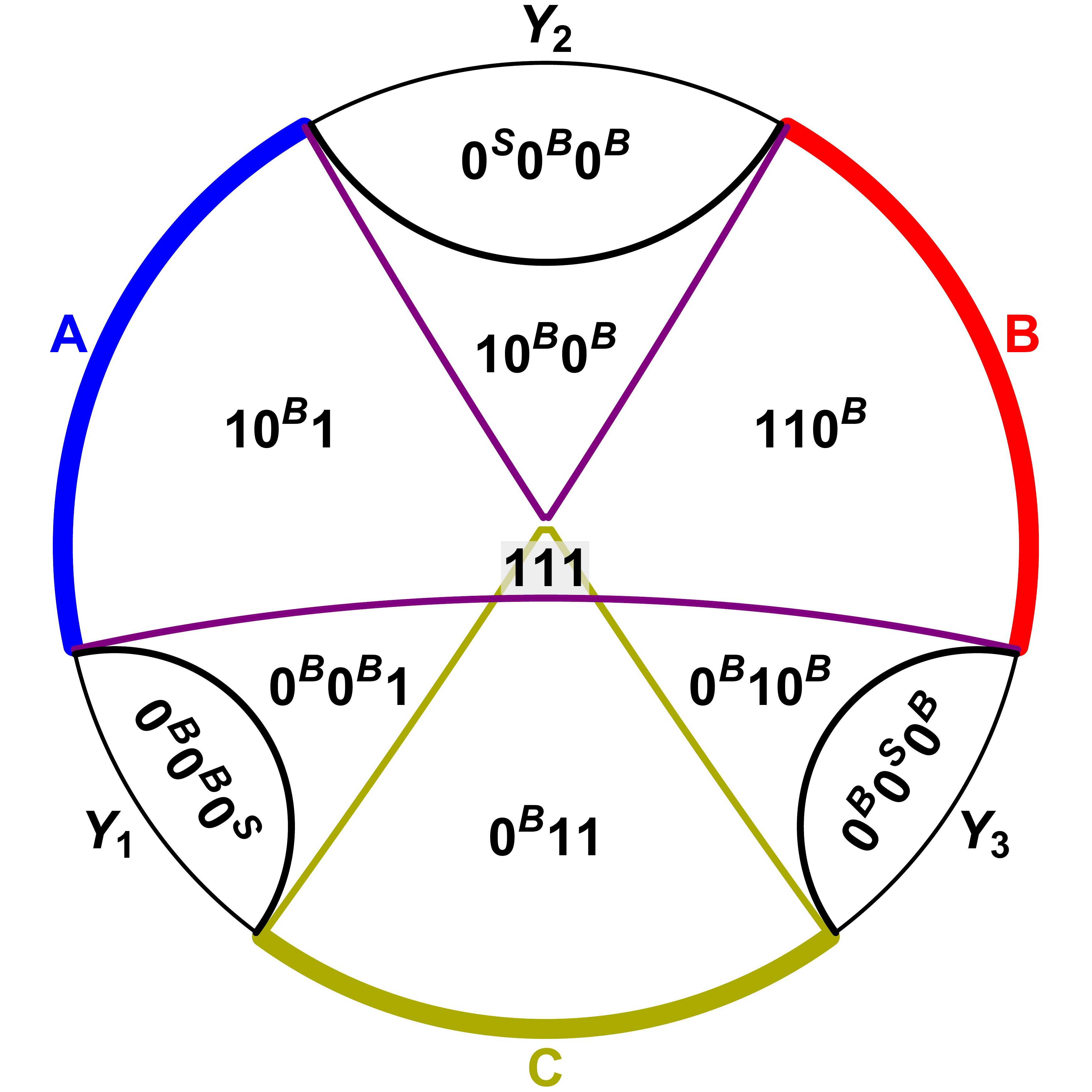}
        \hfill
        \includegraphics[width=0.39\textwidth]{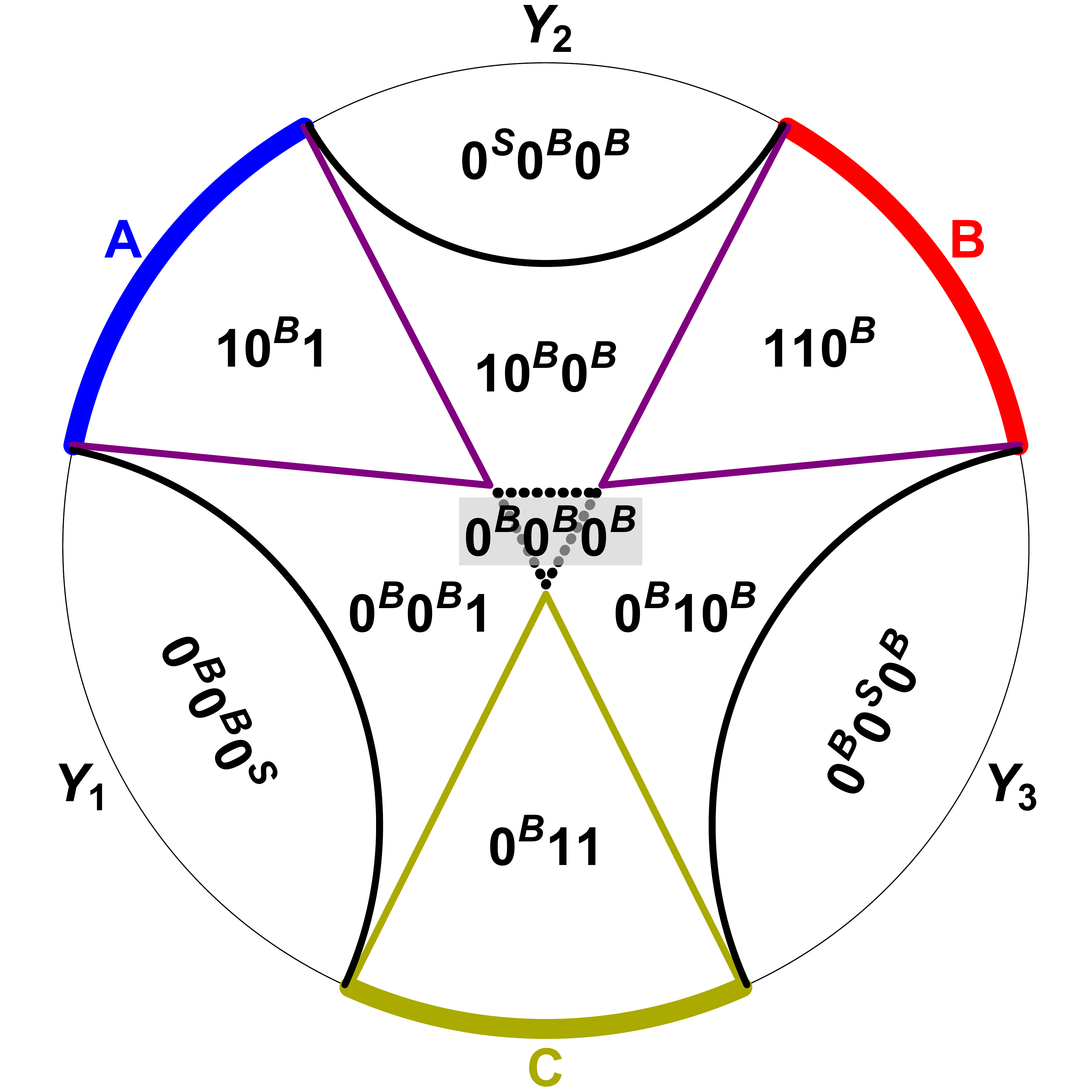}
        \caption{The regions $U_r$ in both auxiliary geometries. The boundaries of $U_{AB}$ are marked in purple, the boundaries of $U_{ABC}$ in black, and the boundaries of $U_C$ in yellow. In Geometry II, $U_{ABC}$ excludes the little central region $0^B0^B0^B$, but this does not affect the coloring of $ABC$; we mark this inconsequential boundary of $U_{ABC}$ in dotted black. Collectively, the boundaries of $U_{AB}$, $U_{ABC}$, and $U_C$ coincide with the boundaries of $W_{AB}$, $W_{BC}$, and $W_{AC}$ from Fig.~\ref{auxgeoms}.}
        \label{auxstatic}
\end{figure}

\paragraph{Comment 2:} Instead, Table~\ref{tab:example} presents a less trivial contraction map, where the strong subadditivity and the tautology `interact.' The fact that the $AB$-column on the range side of $f(\vec{x})$ differs from the $AB$-column on the domain side means that the player can deliberately choose a `suboptimal' coloring for $AB$ and still win the game.

\bigskip

Consider an instance of (\ref{exineq}) where $A, B, C$ are intervals on a circle and $AB, BC, AC$ are each in connected phases. Before using Table~\ref{tab:example} to find a winning strategy, we need to choose an auxiliary geometry. There are two distinct choices, which we call Geometry~I and Geometry~II; see Fig.~\ref{auxgeoms}.

Next, we use Table~\ref{tab:example} to find the regions $U_r$ defined in eq.~(\ref{defujaux}). 
We have drawn an example in Fig.~\ref{auxstatic}. For example, for $U_{AB}$ we have
\begin{align}
U_{AB} 
& = 
(W_{AB} \cap \overline{W_{BC}} \cap W_{AC}) \cup 
(W_{AB} \cap {W_{BC}} \cap \overline{W_{AC}}) \cup 
(W_{AB} \cap {W_{BC}} \cap W_{AC}) 
\nonumber \\
& = W((1,0,1)) \cup W((1,1,0)) \cup W((1,1,1)).
\end{align}

\paragraph{Comment 3:} In Geometry~I, $U_{AB}$ is connected via a bridge region $W((1,1,1))$ so the winning strategy sets $AB$ to the connected phase. But in Geometry~II, the bridge $W((1,1,1))$ is missing and $U_{AB}$ is disconnected. Working with this auxiliary geometry, the player will color $A$ and $B$ in $AB$ with different colors and still win the game---even though she knows the minimal surface for $AB$ is the connected one (as the coloring of $AB$ on the left hand side of the inequality indicates). In the contraction map, this deliberately suboptimal gambit originates from choosing the $AB$-column in $f(\vec{x})$ different from the $AB$-column on the $\vec{x}$ side of the table. 

Importantly, the winning strategies derived from the two auxiliary geometries are different! This illustrates that the winning strategy depends not only on the contraction map $f$, but also on the auxiliary geometry chosen.

\paragraph{Comment 4:} Imagine that our auxiliary geometry (either one) is the spatial slice of a static bulk spacetime. In such a case, drawing the boundaries of the regions $U_r$ on the geometry makes the inequality manifestly true; see Fig.~\ref{auxstatic}. This is because each geodesic segment appears the same\footnote{\ldots or greater, though this does not happen in the present case.} number of times on the left hand side (Fig.~\ref{auxgeoms}) as on the right hand side (Fig.~\ref{auxstatic}). This is the essence of the static proof of Ref.~\cite{hec}. 

In our proof, however, Fig.~\ref{auxstatic} functions only as an auxiliary geometry and it would be sloppy to conclude that inequality~(\ref{exineq}) holds simply by beholding the two figures. Instead, we must interpret the multiplicities of the geodesic segments in Fig.~\ref{auxstatic} as contributions, via eq.~(\ref{pstepwise}), to quantities $h(P, W_l)$ and $h(P, U_r)$ defined in eq.~(\ref{colorfromregion}). 

\bigskip

Before discussing how this works, we remind the reader that in eq.~(\ref{wxcolor}) we have subdivided the regions $W(\vec{x})$ of the auxiliary geometry into color-wise components. Since $W_{AB}$, $W_{BC}$, and $W_{AC}$ are all connected, they only comprise a single color; we omit it in our labeling of Figs.~\ref{auxgeoms} and \ref{auxstatic}. But for their complements, $\overline{W_{AB}}$, $\overline{W_{BC}}$, and $\overline{W_{AC}}$, each comprises two colors, which we label B (bigger) and S (smaller, i.e., a single $Y$-interval). Two examples of colored subregions $W(\vec{x})^{\vec{c}}$ are $W((1,0,1))^{(-,B,-)}$ and $W((0,0,1))^{(B,B,-)}$, which in Figs.~\ref{auxgeoms} and \ref{auxstatic} are labeled $10^{\rm B}1$ and $0^{\rm B} 0^{\rm B}1$.

To examine the logic of our proof, let us look at the single-step path $P_i$ in the auxiliary geometry, which jumps from $W((1,0,1))^{(-,B,-)}$ to $W((0,0,1))^{(B,B,-)}$. We have encountered such single-step paths in eqs.~(\ref{pstepwise}) and (\ref{finalsimp}) in Sec.~\ref{winproof}. The geodesic segment that separates $W((1,0,1))^{(-,B,-)}$ from $W((0,0,1))^{(B,B,-)}$ in the auxiliary geometry is part of the boundary of $W_{AB}$, but not part of the boundary of $W_{BC}$ or $W_{AC}$:
\begin{align}
h(P_i, W_{AB}) & = 1 = \big|(1,0,1) - (0,0,1)\big|_{AB}, \\
h(P_i, W_{BC}) & = 0 = \big|(1,0,1) - (0,0,1)\big|_{BC}, \\
h(P_i, W_{AC}) & = 0 = \big|(1,0,1) - (0,0,1)\big|_{AC}.
\end{align}
We can observe this in Fig.~\ref{auxgeoms}; in Table~\ref{tab:example}, this fact is encoded in the component-wise differences between the $\vec{x}$-vectors $(1,0,1)$ and $(0,0,1)$. For the right-hand-side terms of the inequality, we similarly have:
\begin{align}
h(P_i, U_{AB}) & = 1 = \big|f((1,0,1)) - f((0,0,1))\big|_{AB}, \\
h(P_i, U_{ABC}) & = 0 = \big|f((1,0,1)) - f((0,0,1))\big|_{ABC}, \\
h(P_i, U_{C}) & = 0 = \big|f((1,0,1)) - f((0,0,1))\big|_{C}
\end{align}
These quantities are encoded in the component-wise differences between the respective $f(\vec{x})$ vectors in the contraction map. In Fig.~\ref{auxstatic}, they encode whether or not our geodesic segment is part of the boundary of $U_{AB}$ (respectively $U_{ABC}$ and $U_C$). The number of times the said geodesic segment is swept on Fig.~\ref{auxgeoms} minus the number of times it is swept on Fig.~\ref{auxstatic} is
\begin{equation}
\label{exverification}
\sum_{l=1}^L \alpha_l\, h(P_i, W_l) - \sum_{r=1}^R \beta_r\, h(P_i, U_r),
\end{equation}
which is non-negative by the contraction property. Because the above facts can be read off from Table~\ref{tab:example}, they hold irrespective of which auxiliary geometry we choose.

In Lemma~\ref{ineqks} of Sec.~\ref{ineqoverlap} we established that inequality~(\ref{exineq}) will follow if the overlap functions on kinematic space can be shown to satisfy
\begin{equation}
h_{LHS}(Z_j, Z_k) \equiv \sum_{l=1}^L \alpha_l\, h_l(Z_j, Z_k) \geq 
\sum_{r=1}^R \beta_r\, h_r(Z_j, Z_k) \equiv  h_{RHS}(Z_j, Z_k)
\end{equation}
for all $(Z_j, Z_k) \in \mathcal{K}$. In Lemma~\ref{shortpaths} of Sec.~\ref{winproof}, we saw that for a given $(Z_j, Z_k) \in \mathcal{K}$ a path $P$ in the auxiliary geometry can be found such that 
\begin{eqnarray}
h(P, W_l) = & h_l(Z_j, Z_k) & \quad \textrm{for all}~1 \leq l \leq L \qquad {\rm and} 
\label{hpwlhlzjzk} \\
h(P, U_r) \geq & h_r(Z_j, Z_k) & \quad \textrm{for all}~1 \leq r \leq R.
\end{eqnarray}
This reduced the proof to showing
\begin{equation}
\sum_{l=1}^L \alpha_l\, h(P, W_l) - \sum_{r=1}^R \beta_r\, h(P, U_r) \geq 0
\end{equation}
for some collection of paths $P$. Observing that each such path consists of single-jump paths $P_i$ further reduced the proof to verifying
\begin{equation}
\sum_i \left(\sum_{l=1}^L \alpha_l\, h(P_i, W_l) - \sum_{r=1}^R \beta_r\, h(P_i, U_r)\right) \geq 0.
\label{sumhpiwl}
\end{equation}
The quantity in parentheses is non-negative by the contraction property. This is exactly what we verified in (\ref{exverification}) for the path $P_i$, which jumps from $W((1,0,1))^{(-,B,-)}$ to $W((0,0,1))^{(B,B,-)}$.

\begin{figure}[t]
        \centering
        \includegraphics[width=0.39\textwidth]{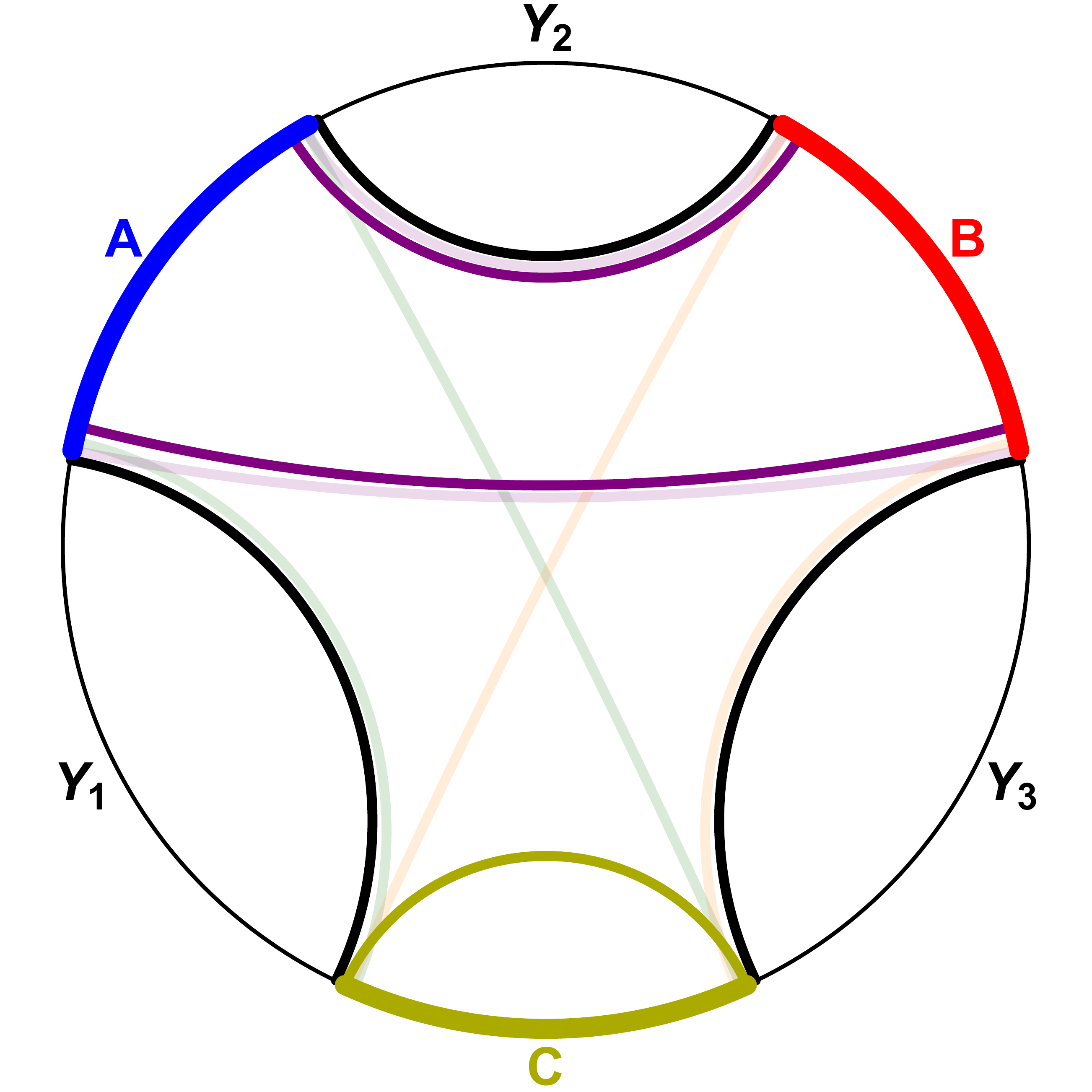}
        \hfill
        \includegraphics[width=0.39\textwidth]{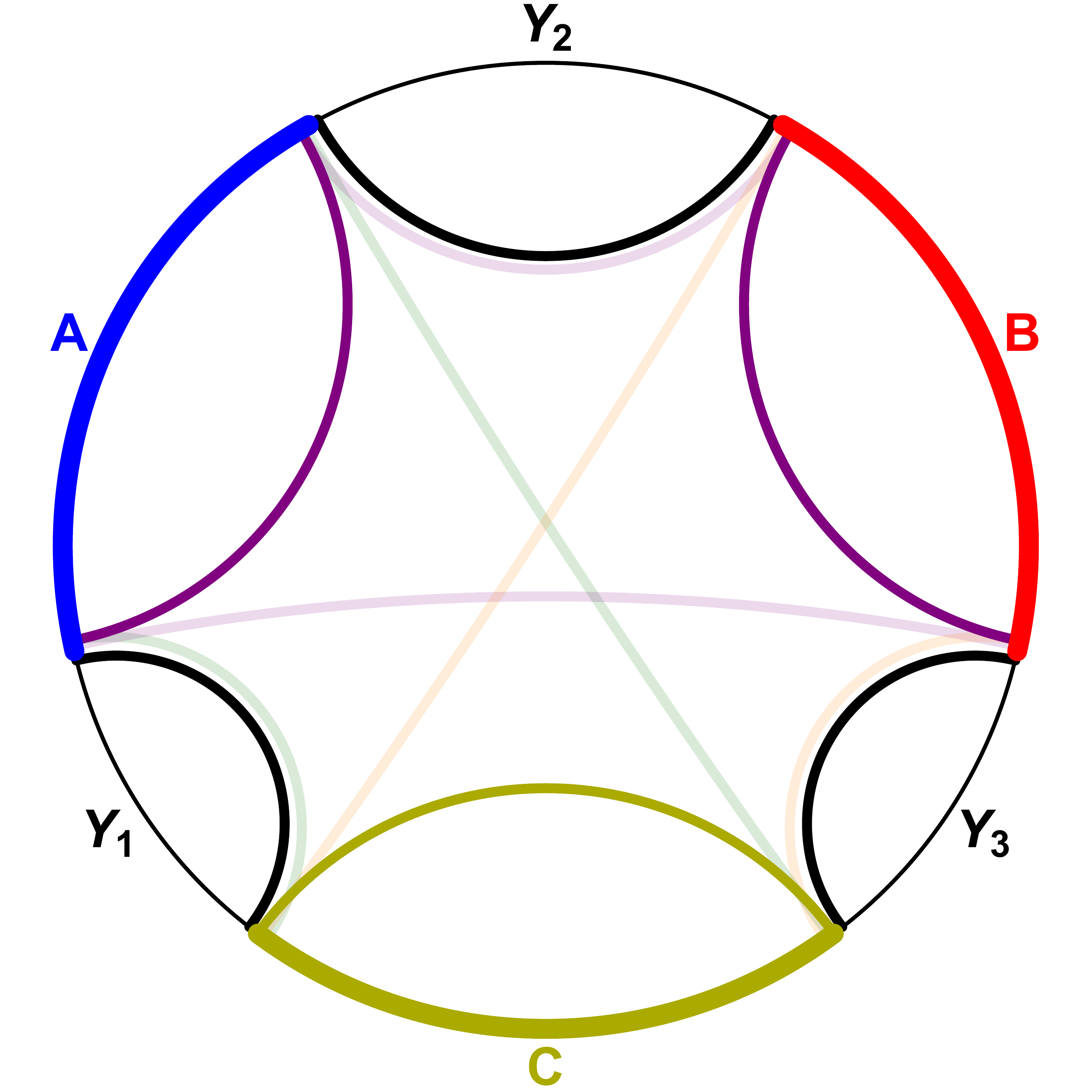}
        \caption{The winning strategies that arise from auxiliary Geometry~I (left) and Geometry~II (right). We superpose them here on top of Geometry~II and Geometry~I, respectively (see Comment~5 for explanation). The color scheme is the same as in Fig.~\ref{auxstatic}.}
        \label{auxfunny}
\end{figure}

\paragraph{Comment 5:} Some readers may still feel puzzled by the freedom to choose an auxiliary geometry. This freedom makes it look like the time-dependent problem is easier than the static problem: not only does the former reduce to the latter but, moreover, the former reduces to an arbitrarily chosen instance of the latter.

To see how this works, let us exercise our freedom to choose an auxiliary geometry in the most contrarian fashion. Suppose that Geometry~I is the physical, spatial slice of a static spacetime, but choose Geometry~II as the auxiliary geometry. We will also look at the opposite situation. We display the winning choices of right hand side geodesics in Fig.~\ref{auxfunny}. 

In contrast to Comment~4, we now cannot verify inequality~(\ref{exineq}) simply by comparing Fig.~\ref{auxfunny} with Fig.~\ref{auxgeoms} because the geodesic segments that appear in them do not coincide. Instead, we must go to Lemma~\ref{ineqoverlap} and verify the equivalent inequality
\begin{equation}
\sum_{l=1}^L \alpha_l\, h_l(Z_j, Z_k) \geq 
\sum_{r=1}^R \beta_r\, h_r(Z_j, Z_k) 
\quad \textrm{for all}~(Z_j, Z_k) \in \mathcal{K}\,.
\label{ourineqagain}
\end{equation}
Of course, we know it will hold by following the argument between eqs.~(\ref{hpwlhlzjzk}) and (\ref{sumhpiwl}) in the auxiliary geometry. In that line of reasoning, drawing Fig.~\ref{auxfunny} is superfluous. In order to develop some spacetime intuition, however, we may read off $h_l(Z_j, Z_k)$ from Fig.~\ref{auxgeoms} and $h_r(Z_j, Z_k)$ from Fig.~\ref{auxfunny}. This is an exercise in counting how many geodesics that participate in the left hand side (Fig.~\ref{auxgeoms}) versus the right hand side (Fig.~\ref{auxfunny}) of inequality~(\ref{exineq}) are crossed by a straight line from $Z_j$ to $Z_k$. 

As an example, we write down the quantities $h_l(Y_2, C)$ and $h_r(Y_2, C)$ in Table~\ref{statichls}. Not surprisingly, inequality~(\ref{ourineqagain}) holds for $(Y_2, C) \in \mathcal{K}$, though to explain why this had to be true we would have to abandon Fig.~\ref{auxfunny} and go back to the correct auxiliary geometry. The exercise works out very similarly for other elements of $\mathcal{K}$.

\begin{table}
\centering
\begin{tabular}{p{1.6cm} r p{0.1cm} c p{0.1cm} c}
  \cline{2-6}
  & & & auxiliary Geometry~I  & & auxiliary Geometry~II \\
  \cline{4-4} \cline{6-6}
  \\[-1em]
  &  $h_{AB(l)}(Y_2, C)$ & & 2 & & 2 \\
  & $h_{BC}(Y_2, C)$ & & 1 & & 1 \\
  & $h_{AC}(Y_2, C)$ & & 1 & & 1 \\
    \\[-1em]
    \cline{2-2} \cline{4-4} \cline{6-6}
   \\[-1em]
  &  $h_{AB(r)}(Y_2, C)$ & & 2 & & 0 \\
  & $h_{ABC}(Y_2, C)$ & & 1 & & 1 \\
  & $h_{C}(Y_2, C)$ & & 1 & & 1 \\
    \\[-1em]
    \cline{1-2} \cline{4-4} \cline{6-6}
    \\[-1em]
  \multicolumn{2}{c}{$\sum_{l=1}^L \alpha_l\, h_l -
\sum_{r=1}^R \beta_r\, h_r$} & & 0 & & 2 \\
    \\[-1em]
  \hline
\end{tabular}
\caption{The quantities $h_l(Y_2, C)$ and $h_r(Y_2, C)$ derived from using either auxiliary geometry.}
\label{statichls}
\end{table}

\paragraph{Comment 6:} 
In compiling Table~\ref{statichls}, we never used the assumed staticity of the bulk spacetime and the geometry of its spatial slices. This assumption turned out to be irrelevant. Our proof happens entirely in kinematic space; it is oblivious of the metric properties of the bulk such as whether or not it is static.

\bigskip

We invite the reader to explore other examples---varying the inequalities, the compositions of regions $A, B, C, \ldots$, choices of phases on the left hand side and, finally, choices of auxiliary geometries.

\se{Multiple asymptotic boundaries and horizons}
\label{nonpure}
Up to now we have focused on cases where the CFT is in a pure state on a circle or line. However, our results hold more generally, including cases where the two-dimensional CFT is in a mixed state and/or lives on a disconnected spatial manifold. 

For a mixed state $\rho$, our reasoning carries over so long as the state can be purified:
\begin{equation}
\rho = {\rm tr}_Q |\Psi \rangle \langle \Psi|
\label{purify}
\end{equation}
such that the HRT formula holds in $|\Psi\rangle$. (We denote the purifying region with $Q$ in order to distinguish it from $O$, the $(n+1)^{\rm st}$ region complementary to the $n$ named regions.) One such purification glues in the bulk the entanglement wedge dual to $\rho$ \cite{dualofrho, subregionduality} with its CPT conjugate taken about the HRT surface where the entanglement wedge ends. This way of constructing a global spacetime dual to a pure state was employed for example in \cite{nettaaron,Dutta:2019gen}.

What remains is to verify that our proof extends to holographic pure states on disconnected spatial manifolds. Recall that such states are dual to bulk spacetimes with multiple asymptotic boundaries, which may be connected or disconnected. In the latter case, we would carry out our proof connected component by connected component. Therefore, in what follows, we will assume that the bulk geometry is connected, i.e. that it represents a multi-boundary wormhole with nontrivial topology. A key challenge in extending our proof to this setup is to enforce the homology condition for the HRT surfaces of various regions. We will do this by lifting the problem from the bulk geometry $\mathcal{B}$ to its universal cover.

\subsection{Proof in multi-boundary wormholes via the universal cover}

We will denote the universal cover of $\mathcal{B}$ with $\tilde{\mathcal{B}}$. Recall that $\mathcal{B} = \tilde{\mathcal{B}} / \pi_1(\mathcal{B})$, i.e., the bulk geometry is the quotient of its universal cover by its fundamental group.  
By definition, $\tilde{\mathcal{B}}$ is topologically trivial so the techniques of our proof of Sec.~\ref{strategy} should carry over to the universal cover. Only one step requires further justification and we comment on it presently. 

We shall lift every ingredient of the problem in $\mathcal{B}$ to $\pi_1(\mathcal{B})$-invariant objects in $\tilde{\mathcal{B}}$. Every geodesic in $\mathcal{B}$ is lifted to a $\pi_1(\mathcal{B})$-invariant set of geodesics in $\tilde{\mathcal{B}}$. (Note that the whole set is invariant, but it consists of geodesics that are not necessarily individually $\pi_1(\mathcal{B})$-invariant.) Furthermore, every boundary region $I_l \subset \mathcal{B}$ lifts to a $\pi_1(\mathcal{B})$-invariant region $\tilde{I_l} \subset \tilde{\mathcal{B}}$. Observe that if a collection of geodesics in $\mathcal{B}$ is homologous to $I_l$ then their lift in $\tilde{\mathcal{B}}$ is homologous to $\tilde{I_l}$. Roughly speaking, our tactic will be to prove inequality~(\ref{templateineq}) by proving
\begin{equation}
\sum_{l=1}^L \alpha_l S(\tilde{I_l}) 
\geq 
\sum_{r=1}^R \beta_r S(\tilde{J_r})\,.
\label{templateineq3}
\end{equation}

The subtlety alluded to earlier is that our proof relies on the strong subadditivity property of geodesics that subtend contiguous intervals. This is most evident in eq.~(\ref{ineqisssa}). The satisfaction of a similar inequality by geodesics in $\tilde{\mathcal{B}}$  cannot, on the face of it, be directly inferred from strong subadditivity in $\mathcal{B}$ because a (minimal) geodesic in $\tilde{\mathcal{B}}$ may well be a lift of a non-minimal geodesic in $\mathcal{B}$, i.e., one that does not compute an entanglement entropy in the CFT on the asymptotic boundary of $\mathcal{B}$. It turns out that this is in fact not a problem. In Sec.~\ref{strategy} we could assume that the strong subadditivity statement holds for minimal geodesics in $\mathcal{B}$ because Ref.~\cite{maximin} derived it from the null curvature condition in $\mathcal{B}$. But if the null curvature condition holds in $\mathcal{B}$ then it also holds in $\tilde{\mathcal{B}}$ and the derivation of \cite{maximin} can be carried out directly in $\tilde{\mathcal{B}}$. This will produce a strong subadditivity type statement for lengths of all geodesics in $\tilde{\mathcal{B}}$, regardless of their minimality in $\mathcal{B}$. 

As a first step in the proof, consider a collection of geodesics in $\mathcal{B}$, which together satisfy the homology conditions for regions $I_l$ on the left hand side of inequality~(\ref{templateineq}). This choice represents a phase of the left hand side and, equivalently, one possible coloring of the $I_l$s. We then lift this collection of geodesics to $\tilde{\mathcal{B}}$ to obtain a phase of the left hand side of (\ref{templateineq3}). The coloring of the $\tilde{I_l}$s  obtained in this way is not any generic coloring; it is a $\pi_1(\mathcal{B})$-invariant coloring defined below.

\paragraph{Definition} A coloring of $\td I$ is $\pi_1(\mathcal{B})$-invariant when the following holds: for any two basic intervals $Z_j, Z_k \subset \td I$ that have the same color and for every $g\in \pi_1(\mathcal{B}),$ the two basic intervals $g(Z_j)$, $g(Z_k)$ also have the same color.

It is straightforward to show that (as we asserted above) the lift $I_l \to \tilde{I_l}$ lifts a coloring of $I_l$ to a $\pi_1(\mathcal{B})$-invariant coloring of $\tilde{I_l}$. Observe that the converse is also true: a $\pi_1(\mathcal{B})$-invariant coloring of $\tilde{J_r}$ in $\tilde{\mathcal{B}}$ descends to a well-defined coloring of $J_r$ in $\mathcal{B}$.

\bigskip
With these preliminaries, a proof of inequality~(\ref{templateineq}) from inequality~(\ref{templateineq3}) reduces to demonstrating the following claim: For every $\pi_1(\mathcal{B})$-invariant coloring of the regions $\tilde{I_l} \subset \tilde{\mathcal{B}}$, we can find a $\pi_1(\mathcal{B})$-invariant coloring of the $\tilde J_r$s whose total geodesic length is no greater. In effect, we must prove inequality~(\ref{templateineq3}) with the extra proviso that the $S(\tilde{I_l})$s and $S(\tilde{J_r})$s are given by colorings that are $\pi_1(\mathcal{B})$-invariant. A proof of (\ref{templateineq3}) with this extra stipulation suffices because a $\pi_1(\mathcal{B})$-invariant coloring of the $\tilde{J_r}$s descends to a coloring of $J_r$s and, as a consequence, inequality~(\ref{templateineq3}) in $\tilde{\mathcal{B}}$ is nothing but $|\pi_1(\mathcal{B})|$ copies of the original inequality~(\ref{templateineq}) in $\mathcal{B}$. 

Working in the topologically trivial $\tilde{\mathcal{B}}$, we proceed with proving (\ref{templateineq3}) just as we did in Sec.~\ref{strategy}. The only place where our proof allowed an arbitrary choice was the selection of the auxiliary geometry. In order to obtain a $\pi_1(\mathcal{B})$-invariant winning coloring of the $\tilde{J_r}$s, we choose a $\pi_1(\mathcal{B})$-invariant auxiliary geometry. 

To confirm the $\pi_1(\mathcal{B})$-invariance of the winning strategy, observe that the regions $W_l$ in the auxiliary geometry of $\tilde{\mathcal{B}}$ are themselves $\pi_1(\mathcal{B})$-invariant, and so are their complements. Therefore, according to eqs.~\er{defwx} and~\er{defujaux}, $W(\vec x)$ and $U_r$ are also $\pi_1(\mathcal{B})$-invariant regions in the auxiliary geometry. Recall that two basic intervals $Z_j, Z_k \subset \td J_r$ have the same color in $\td J_r$ if and only if they belong to the same connected component of $U_r$. Since $U_r$ is $\pi_1(\mathcal{B})$-invariant, if $Z_j, Z_k \subset \td J_r$ have the same color then so must $g(Z_j)$ and $g(Z_k)$, for every $g \in \pi_1(\mathcal{B})$. This confirms that $U_r$ induces a $\pi_1(\mathcal{B})$-invariant coloring of $\tilde{J_r}$ and completes our proof. 

\subsection{Example: BTZ black hole}

As an illustration of the concepts used above, imagine proving an inequality~(\ref{templateineq}) in the one-sided non-rotating BTZ geometry, which is dual to a thermal state on a circle. Since the thermal state is mixed, as a first step we will purify the state. Gluing the single-sided BTZ geometry to its CPT conjugate in the bulk produces the two-sided BTZ black hole, which will be our geometry $\mathcal{B}$. In the CFT language, we are applying here the channel-state duality to obtain the thermofield double state $|\Psi\rangle$. (In the matrix notation, we are mapping rows of the density matrix $\rho$ to column vectors in another copy of the Hilbert space.) This second Hilbert space, which we denoted $Q$ in equation~(\ref{purify}), lives on the second asymptotic boundary of $\mathcal{B}$.  

\begin{figure}
        \centering
        \raisebox{0.08\textwidth}{\includegraphics[width=0.25\textwidth]{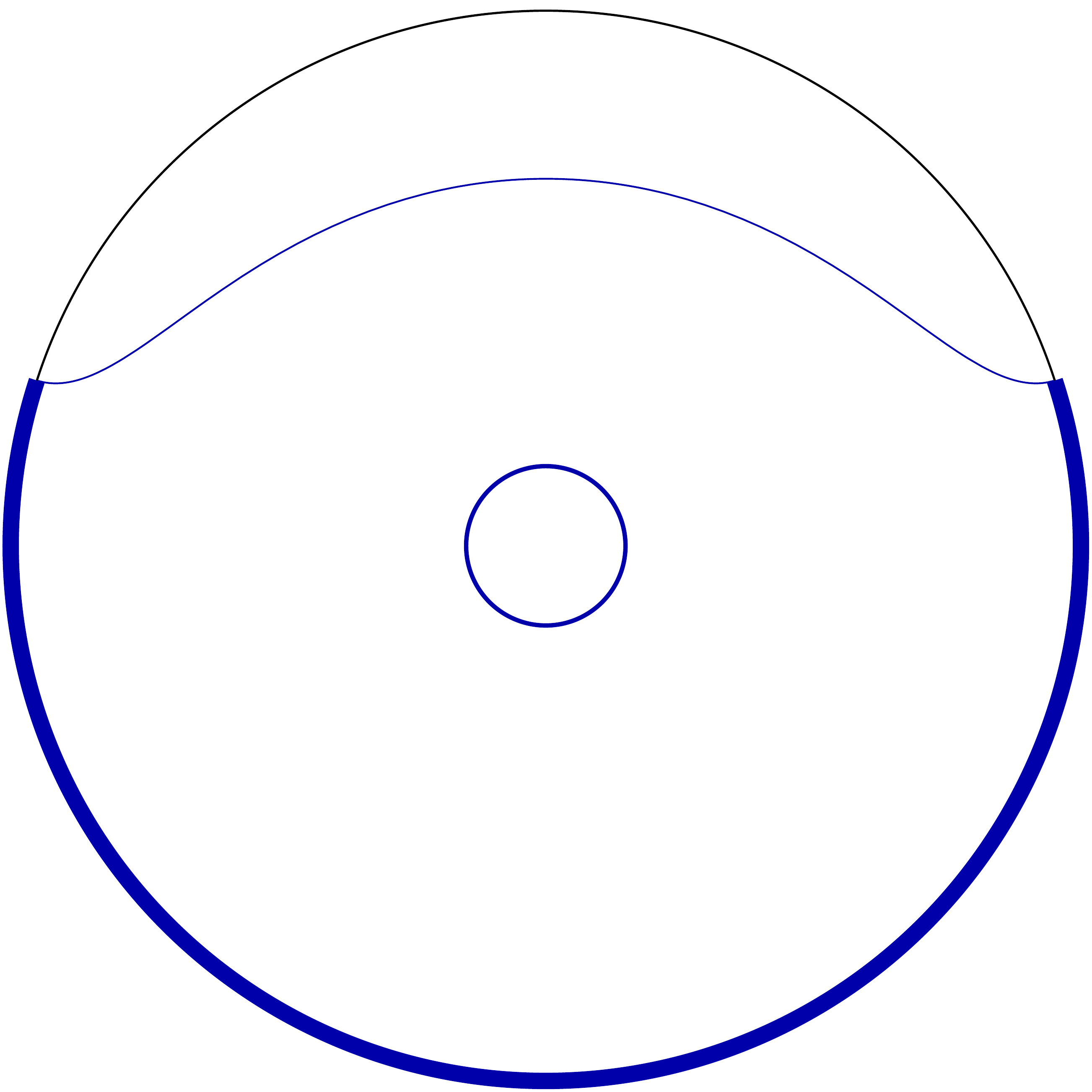}}
     \hfill
        \raisebox{0.05\textwidth}{\includegraphics[width=0.27\textwidth]{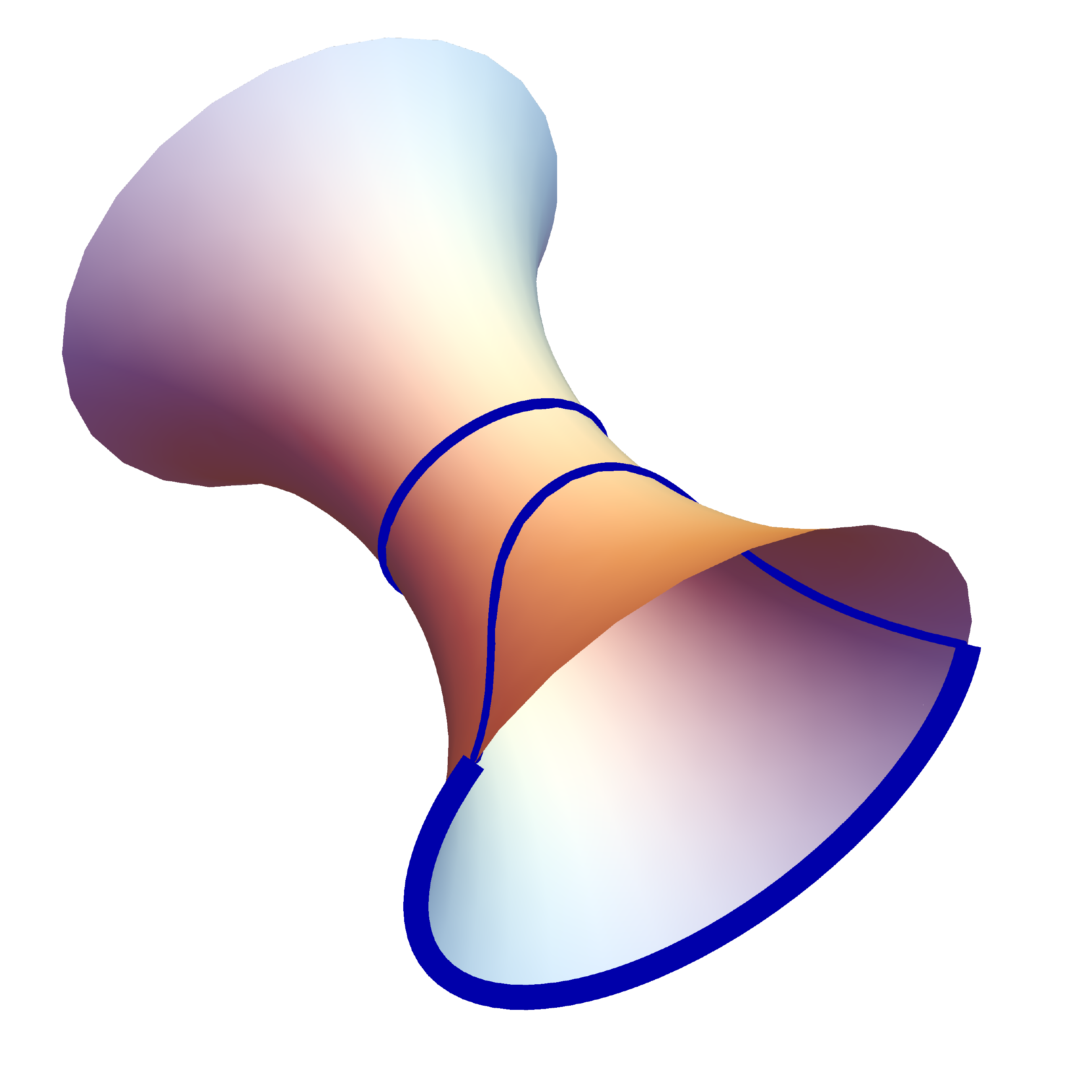}}
        \hfill
        \includegraphics[width=0.40\textwidth]{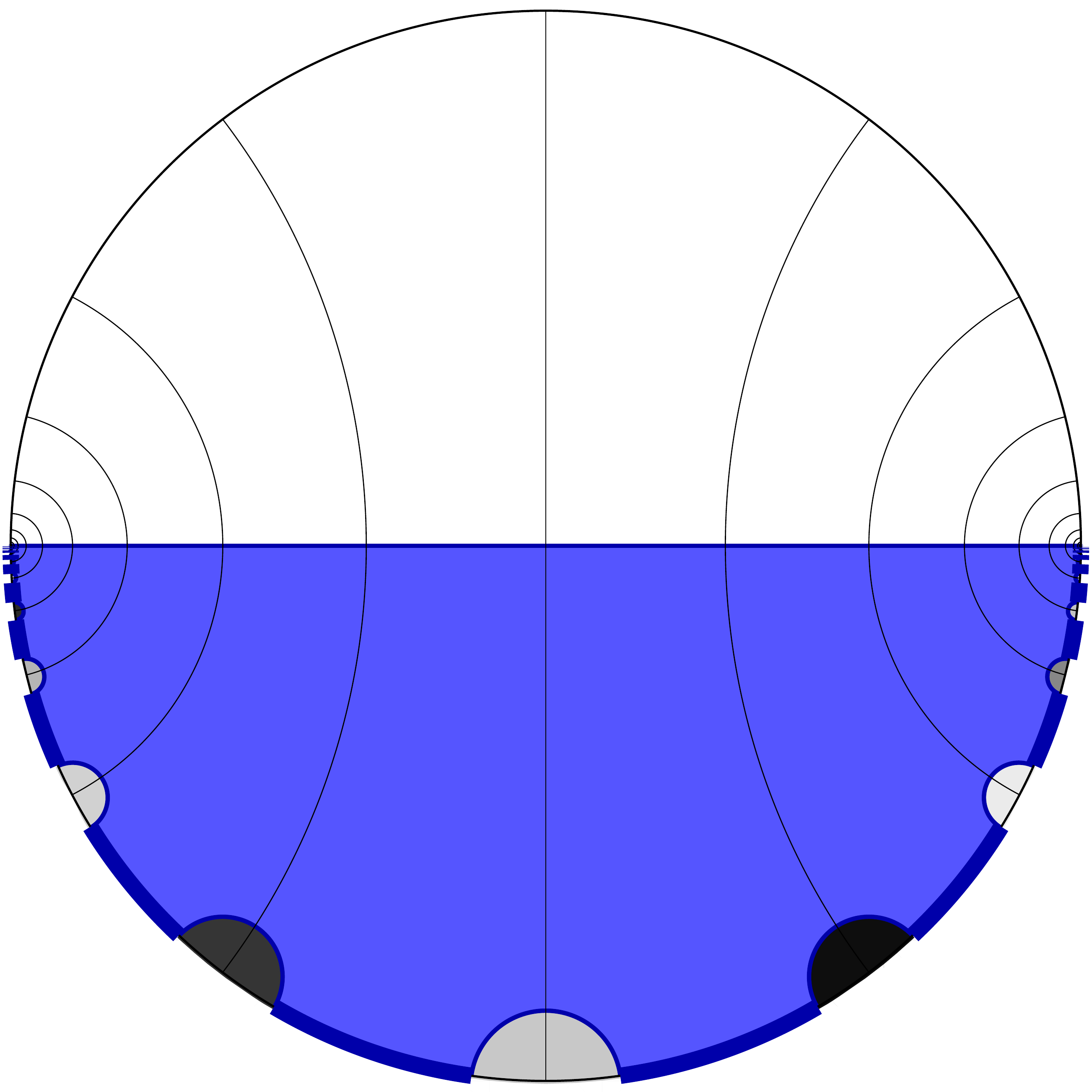} 
   \caption{Left: the one-sided BTZ geometry. There are two types of intervals, whose HRT surfaces either consist of a single geodesic (`small' intervals, here non-highlighted) or include the black hole horizon (`large' intervals, highlighted). In our proof, we first purify the state by gluing it with its CPT conjugate (the two-sided BTZ geometry, center panel) and then go to its universal cover (a spatial slice of AdS$_3$, right panel). The images of a small interval end up in the disconnected phase (colored with diverse shades of gray); the images of a large interval are in the connected phase, one of whose geodesics is the concatenation of the images of the black hole horizon.}
   \label{bhex}
\end{figure}

As we illustrate in Fig.~\ref{bhex}, intervals on the boundary of $\mathcal{B}$ come in two varieties. Their holographic entanglement entropies are realized either by a geodesic that is individually homologous to the interval or by a geodesic that wraps around the black hole horizon plus the horizon itself. Although the physically realized phase of the entanglement entropy is decided both by the size of the interval and by the relative boost between the interval and the static slicing, we shall refer to the two classes of intervals simply as `small' and `large.' 

While geometry $\mathcal{B}$ is dual to a pure state, it is not topologically trivial; its fundamental group is $\pi_1(\mathcal{B}) \cong \mathbb{Z}$. We will therefore go to its universal cover $\tilde{\mathcal{B}}$, which is the AdS$_3$ geometry. (In this case $\pi_1(\mathcal{B})$ is a subgroup of the group of continuous isometries of $\tilde{\mathcal{B}}$, but this is an artifact of the BTZ example.) 
Let us inspect how the small and large intervals, as well as their bulk extremal surfaces, lift to $\tilde{\mathcal{B}}$.

Each type of interval lifts to a discretely infinite family of intervals, which are labeled by elements of $\pi_1(\mathcal{B}) \cong \mathbb{Z}$. The difference between small and large intervals comes at the level of their coloring. For a small interval $Z$, each image $\tilde Z$ in $\tilde{\mathcal{B}}$ comes in its own color and the $\tilde{Z}$s are in the totally disconnected phase. Such a coloring is $\pi_1(\mathcal{B})$-invariant; the invariance of a coloring only stipulates that $Z_j \sim Z_k \Rightarrow g(Z_j) \sim g(Z_k)$ for all $g \in \pi_1(\mathcal{B})$ but not $g(Z_j) \sim h(Z_k)$ for all $g, h \in \pi_1(\mathcal{B})$.

In contrast, all images $\tilde{Z'}$ of a large interval $Z'$ end up in the same color in $\tilde{\mathcal{B}}$, i.e. the $\tilde{Z'}$s are in the totally connected phase. The pieces of the extremal surfaces in $S(Z')$ that run alongside the black hole horizon in $\mathcal{B}$ come together in $\tilde{\mathcal{B}}$ to form a single geodesic, which connects the two accumulation points of the infinitely many images of the $Z'$s in $\tilde{\mathcal{B}}$. 

As promised, the $\tilde{\mathcal{B}}$ lifts of the HRT surfaces are homologous to the lifts of the intervals. Therefore, we can prove inequality~(\ref{templateineq3}) by following the argument of Sec.~\ref{strategy}. This inequality will be a countably infinite multiple of the original inequality~(\ref{templateineq}) in $\mathcal{B}$, as should be evident from the right panel of Fig.~\ref{bhex}.


\se{Discussion}
\label{discussion}
Let us step back and ask why our proof works---and why it does not obviously generalize to higher dimensions. As the time-dependent problem mandates, we have not referenced any particular slice of the bulk geometry. Furthermore, by working with the discretized version of the kinematic space that bins together entire intervals (Fig.~\ref{fig:kspace} in Sec.~\ref{sec:ks}), we have essentially stripped the problem off geometric data and reduced it to a topological one. Indeed, the intersection numbers that we manipulated are topological quantities, which originate from the ordering of intervals on a circle or line. Because this structure is special to two boundary dimensions, our proof is unlikely to extend to higher-dimensional cases in its current form.

We can make this point a little more explicit.
To prove inequalities~(\ref{templateineq}), we equate the difference between the left hand side and the right hand side with a combination of conditional mutual information quantities. (This is most clearly visible from eq.~(\ref{ineqisssa}).) Evidently, in two boundary dimensions, the strong subadditivity of contiguous intervals $A, B, C$
\begin{equation}
I(A, C | B) = S(AB) + S(BC) - S(ABC) - S(B) \geq 0
\label{ssacont}
\end{equation}
determines the entire holographic entropy cone. In the CFT computation of entanglement entropies \cite{tomsentanglement},  inequality~(\ref{ssacont}) selects the channel in which the identity conformal block dominates the correlation function of four twist operators. In this language, repeated applications of strong subadditivity select a dominant channel in the calculation of a higher-point correlation function of twist operators. This is a combinatorial problem, which our proof solves by topological means. In higher dimensions, where entanglement entropies are not computed by correlation functions of local operators on an ordered space, the problem is not combinatorial in nature and it is unlikely to have a topological solution.

It is generally expected that the facets of the holographic entropy cone should correspond to special entanglement structures that afford holographic interpretations. From Refs.~\cite{bitthreadmonogamy, seperstab} we know that one such ingredient is the perfect tensor entanglement, which is associated with the monogamy of mutual information. The other holographic entropy inequalities should reveal further atomic ingredients from which holographic spacetimes are built \cite{entropyrelations}. This suggests that a general proof of the holographic entropy cone with time dependence may require a novel approach that is aware of these other entanglement structures without using special features of three bulk dimensions.

\section*{Acknowledgments}
We thank Ning Bao, Sa{\v s}o Grozdanov, Daniel Harlow, Sergio Hern\'andez Cuenca, Veronika Hubeny, Donald Marolf, Hirosi Ooguri, Mukund Rangamani, Massimiliano Rota, Wilke van der Schee, Bogdan Stoica, Gabriel Trevi{\~n}o Ver{\'a}stegui, Aron Wall, Michael Walter, and Wayne Weng for useful discussions.  BC is supported by the Young Thousand Talents Program and by a Start-up Fund from Tsinghua University. BC thanks the Institute for Advanced Study (Princeton), the University of Pennsylvania, MIT, the Weizmann Institute, Fudan University, and AEI Potsdam, where part of this work was completed. XD is supported in part by the National Science Foundation under Grant No.\ PHY-1820908 and by funds from the University of California.  XD is grateful to the Institute for Advanced Study and the KITP where part of this work was developed.  The KITP was supported in part by the National Science Foundation under Grant No.\ PHY-1748958.

\appendix
\se{Minimal order within a color respects the spatial ordering}
\label{apporder}

Let $X_i$ be a collection of intervals whose indices $i$ respect the spatial ordering on a line or circle (the CFT spatial slice). We claim that, in the notation of eq.~(\ref{phasescycles}), the total length of geodesics in the phase $(X_1 X_2 \cdots X_k)$ \emph{in a pure state} is no greater than the total length in any putative phase $(X_{\sigma_1} X_{\sigma_2} \cdots X_{\sigma_k})$, where $\sigma$ is a permutation of $k$ elements. Therefore, the latter phase can never be the one that determines the entanglement entropy of any CFT region in the HRT proposal (unless $\s$ is the identity). To avoid clutter, we will abuse the notation of eq.~(\ref{phasescycles}) slightly and simply write our claim as
\begin{equation}
(X_{\sigma_1} X_{\sigma_2} \cdots X_{\sigma_k}) \geq (X_1 X_2 \cdots X_k)\,.
\label{minorder}
\end{equation}
That is, in this appendix we will use the same cycle notation to mark the phase and the total length of geodesics in that phase.

This claim follows from the strong subadditivity property of entanglement entropies. For later convenience, we will write explicitly two special instances of strong subadditivity. Consider four points $x_i$ ($i = 1, 2, 3, 4$) on a line or circle and assume that the index $i$ respects their spatial ordering. For $i < j$, let $(x_i, x_j)$ denote the interval\footnote{We do not distinguish between open and closed intervals, so $(x_1, x_2) \cup (x_2, x_3) = (x_1, x_3)$.} whose left endpoint is $x_i$ and right endpoint is $x_j$. On a line, we will supplant this standard notation with a counterpart for complements of intervals: for $i > j$, let $(x_i, x_j) \equiv (x_i, +\infty) \cup (-\infty, x_j)$. This definition matches the notation on a circle. Then the following inequalities are both instances of strong subadditivity:
\begin{align}
I\big(x_2 (x_3 x_4)x_1\big) \equiv S\big( (x_2, x_4) \big) + S\big( (x_1, x_3) \big) - S\big( (x_2, x_3) \big) - S\big( (x_1, x_4) \big) \geq 0, \nonumber \\
I\big(x_1(x_2 x_3)x_4\big) \equiv S\big( (x_1, x_3) \big) + S\big( (x_4, x_2) \big) - S\big( (x_1, x_2) \big) - S\big( (x_4, x_3) \big) \geq 0.
\label{ssaexplicit}
\end{align}
They are both of the form
\begin{equation}
I(A, C | B) \equiv S(AB) + S(BC) - S(B) - S(ABC) \geq 0\,,
\end{equation}
i.e., both represent the positivity of conditional mutual information $I(A, C | B)$, with:
\begin{align}
A = (x_1, x_2) \quad B = (x_2, x_3) \quad C=(x_3, x_4) \qquad & \textrm{in the first case and} \\
A = (x_4, x_1) \quad B = (x_1, x_2) \quad C=(x_2, x_3) \qquad & \textrm{in the second case.}
\end{align}
On a circle, there is no difference between the two cases but on a line they differ in that the `interval' $A$ includes infinity in the second case. In summary, expression $I\big(x_i (x_j x_k) x_l\big)$ is non-negative whenever $x_i, x_j, x_k, x_l$ form an ordered sequence on a circle (i.e., modulo a cyclic re-mapping).

We will prove eq.~(\ref{minorder}) by induction on the number of intervals $k$. At $k=1$ and $k=2$, the claim holds trivially. 
Assume that inequality~(\ref{minorder}) is true on $k-1$ intervals. It suffices to prove that:
\be
(X_{\s_1} X_{\s_2} \cd X_{\s_k}) - (X_1 X_2 \cd X_k) \ge (X_{\s_1} \hat X_{\s_2} \cd X_{\s_k}) - (X_1 X_2 \cd \hat X_{\s_2} \cd X_k),
\ee
where $\hat .$ stands for omitting the interval from the cycle. Now, because we are working in a pure state,  we recognize that
\begin{align}
(X_{\s_1} X_{\s_2} \cd X_{\s_k}) \! - \! (X_{\s_1} \hat X_{\s_2} \cd X_{\s_k})
\! & = \! S\big( (X_{\sigma_1}^R, X_{\sigma_2}^L) \big) 
\! + \! S\big( (X_{\sigma_2}^R, X_{\sigma_3}^L) \big) 
\! - \! S\big( (X_{\sigma_1}^R, X_{\sigma_3}^L) \big),\nonumber\\
(X_1 X_2 \cd X_k)\! -\! (X_1 X_2 \cd \hat X_{\s_2} \cd X_k) 
\! & = \! S\big( (X_{\sigma_2-1}^R, X_{\sigma_2}^L) \big) 
\! + \! S\big( (X_{\sigma_2}^R, X_{\sigma_2+1}^L) \big) 
\! - \! S\big( (X_{\sigma_2-1}^R, X_{\sigma_2+1}^L) \big) 
\end{align}
where $X^L$ and $X^R$ denote the left and right endpoints of interval $X$. Thus, our task reduces to showing that:
\bm
S\big( (X_{\sigma_2}^R, X_{\sigma_3}^L) \big)
- S\big( (X_{\sigma_1}^R, X_{\sigma_3}^L) \big) 
+ S\big( (X_{\sigma_1}^R, X_{\sigma_2}^L) \big) \\
- \, S\big( (X_{\sigma_2}^R, X_{\sigma_2+1}^L) \big)
+ S\big( (X_{\sigma_2-1}^R, X_{\sigma_2+1}^L) \big) 
- S\big( (X_{\sigma_2-1}^R, X_{\sigma_2}^L) \big) \geq 0.
\label{finalcrosses}
\em
This quantity is a combination of conditional mutual information quantities, which we saw to be nonnegative in inequalities~(\ref{ssaexplicit}). To see this, we distinguish two cases which depend on the ordering of $X_{\sigma_1}, X_{\sigma_2}$, and $X_{\sigma_3}$.

If $\sigma_1 < \sigma_2 < \sigma_3$, $\sigma_2 < \sigma_3 < \sigma_1$, or $\sigma_3 < \sigma_1 < \sigma_2$, the left hand side of (\ref{finalcrosses}) equals
\begin{align}
& S\big( (X_{\sigma_2-1}^R, X_{\sigma_2+1}^L) \big) + S\big( (X_{\sigma_1}^R, X_{\sigma_2}^L) \big)
- S\big( (X_{\sigma_2-1}^R, X_{\sigma_2}^L) \big) - S\big( (X_{\sigma_1}^R, X_{\sigma_2+1}^L) \big) \nonumber \\
+ \, & S\big( (X_{\sigma_2}^R, X_{\sigma_3}^L ) \big) + S\big( (X_{\sigma_1}^R, X_{\sigma_2+1}^L) \big)
- S\big( (X_{\sigma_2}^R, X_{\sigma_2+1}^L \big) - S\big( (X_{\sigma_1}^R, X_{\sigma_3}^L) \big) \nonumber \\
= \, & I\big( X_{\sigma_2-1}^R  (X_{\sigma_2}^L X_{\sigma_2+1}^L) X_{\sigma_1}^R \big) 
+ I\big( X_{\sigma_2}^R (X_{\sigma_2+1}^L X_{\sigma_3}^L) X_{\sigma_1}^R \big) \geq 0.
\end{align}
If $\sigma_1 < \sigma_3 < \sigma_2$, $\sigma_3 < \sigma_2 < \sigma_1$, or $\sigma_2 < \sigma_1 < \sigma_3$, the left hand side of (\ref{finalcrosses}) equals
\begin{align}
& S\big( (X_{\sigma_1}^R, X_{\sigma_2}^L) \big) + S\big( (X_{\sigma_2+1}^L, X_{\sigma_2-1}^R) \big)
- S\big( (X_{\sigma_1}^R, X_{\sigma_2-1}^R) \big) - S\big( (X_{\sigma_2+1}^L, X_{\sigma_2}^L) \big) \nonumber \\
+ \, & S\big( (X_{\sigma_3}^L, X_{\sigma_2}^R ) \big) + S\big( (X_{\sigma_2+1}^L, X_{\sigma_2}^L) \big)
- S\big( (X_{\sigma_3}^L, X_{\sigma_2}^L) \big) - S\big( (X_{\sigma_2+1}^L, X_{\sigma_2}^R \big) \nonumber \\
+ \, & S\big( (X_{\sigma_1}^R, X_{\sigma_2-1}^R) \big) + S\big( (X_{\sigma_2}^L, X_{\sigma_3}^L) \big)
- S\big( (X_{\sigma_1}^R, X_{\sigma_3}^L) \big) - S\big( (X_{\sigma_2}^L, X_{\sigma_2-1}^R) \big) \nonumber \\
= \, & I\big( X_{\sigma_1}^R (X_{\sigma_2-1}^R X_{\sigma_2}^L) X_{\sigma_2+1}^L \big) 
+ I\big( X_{\sigma_3}^L (X_{\sigma_2}^L X_{\sigma_2}^R) X_{\sigma_2+1}^L \big) 
+ I\big( X_{\sigma_1}^R (X_{\sigma_3}^L X_{\sigma_2-1}^R) X_{\sigma_2}^L \big) \geq 0\,.
\end{align}

\bibliographystyle{JHEP}
\bibliography{entropyconebibliography}

\end{document}